\documentclass[a4paper,11pt]{article} \usepackage{jheppub}
\usepackage{mathrsfs}
\usepackage{verbatim}
\usepackage{booktabs}
\usepackage{placeins}
\usepackage{pdflscape}
\usepackage{amsthm}
\usepackage{enumitem}

\newcommand{\R}{\mathbb{R}}

\newcommand{\C}{\mathbb{C}}

\newcommand{\nn}{\nonumber}
\newcommand{\mc}[1]{\mathcal{#1}}

\newcommand{\pd}[2]{\frac{\partial #1}{\partial #2}}
\newcommand{\tree}{\mathrm{tree}}
\newcommand{\avg}[1]{\langle#1\rangle}

\newcommand{\aAvg}[3]{\langle#1|#2|#3]}

\newcommand{\expval}[3]{\langle #1|\hspace*{.2mm}#2\hspace*{.3mm}|#3 \rangle}
\newcommand{\Res}{\mathop{\rm Res}}
\newcommand{\IP}{\mc I^{\mathrm{P}}_{10}}

\newcommand{\cut}{\operatorname{cut}}
\newcommand{\bc}{\begin{center}}
\newcommand{\ec}{\end{center}}

\newcommand{\LRa}{\Longrightarrow}

\newtheorem{thm}{Theorem}
\newtheorem{lemma}{Lemma}
\newtheorem{example}{Example}

\title{Multivariate Residues and Maximal Unitarity}

\author{Mads S{\o}gaard, Yang Zhang}
\affiliation{
Niels Bohr International Academy and Discovery Center, Niels Bohr Institute, \\
University of Copenhagen, Blegdamsvej 17, DK-2100 Copenhagen, Denmark
}
\emailAdd{madss@nbi.dk}
\emailAdd{zhang@nbi.dk}
\abstract{We extend the maximal unitarity method to amplitude contributions
whose cuts define multidimensional algebraic varieties. The technique is valid
to all orders and is explicitly demonstrated at three loops in gauge theories
with any number of fermions and scalars in the adjoint representation.
Deca-cuts realized by replacement of real slice integration contours by
higher-dimensional tori encircling the global poles are used to factorize the
planar triple box onto a product of trees. We apply computational algebraic
geometry and multivariate complex analysis to derive unique projectors for all
master integral coefficients and obtain compact analytic formulae in terms of
tree-level data.}

\begin{document}
\maketitle
\flushbottom
\clearpage

\section{Introduction}
Potential discovery of new physics depends on our ability to compute precision
cross section predictions for scattering of subatomic particles and in
particular, a quantitative understanding of all relevant Standard Model
processes which necessarily must be separated from the experimental data.
Theoretical calculations carried out at the LHC start with tree-level
amplitudes at leading order (LO) in perturbative Quantum Chromodynamics (QCD),
whereas a combination of one-loop amplitudes and higher-multiplicity trees
provide next-to-leading order corrections of quantitative reliability.
Although computations of this type are very complicated, recent years have
seen major advances at NLO, especially for processes with many final states.
In the near future, theoretical calculations offered at NLO become
insufficient to saturate the bound for experimental uncertainty. The upcoming
frontier is therefore formed by NNLO computations and in particular,  two-loop
amplitudes, which are also relevant already at NLO for specific processes such
as production of electroweak gauge bosons by gluon fusion, for example. 

Historically, contributions to perturbative scattering amplitudes have been
tracked pictorially by means of Feynman diagrams, which lead to precise
mathematical expressions using the Feynman rules. Although this method gives
an indispensable fundamental intuition for interactions between elementary
particles, even simple problems beyond two-by-two gluon scattering beyond the
tree-level become cumbersome due to the presence of a large redundancy in the
theory needed in order to compensate for virtual intermediate states. In the
last decade, strikingly powerful on-shell methods for amplitude computations
at tree- and loop-level involving only physical information and analytic
properties have emerged, the most prominent examples being the
Britto-Cachazo-Feng-Witten (BCFW) \cite{Britto:2004ap,Britto:2005fq} recursion
relations and the generalized unitarity method, an enhanced version of the
original unitarity method due to Bern, Dixon, and Kosower
\cite{Bern:1994cg,Bern:1994zx}. Their importance is reflected by the fact that
all gauge theory and also gravity trees now may be constructed by  
just the Cauchy residue theorem and complex kinematics in three-point
amplitudes, which then are recycled for loops. 

The unitarity method (see also
\cite{Bern:1995db,Bern:1997sc,Britto:2004nc,Britto:2004nj,
Bern:2005hh,Bidder:2005ri,Britto:2005ha,Britto:2006sj,Mastrolia:2006ki,
Brandhuber:2005jw,Ossola:2006us,Anastasiou:2006gt,Bern:2007dw,Forde:2007mi,
Badger:2008cm,Giele:2008ve,Britto:2006fc,Britto:2007tt,Bern:2010qa,
Anastasiou:2006jv} for extensive subsequent studies) exploits that the
discontinuity of the transition matrix can be expressed in terms of simpler
quantities to probe the analytic structure of the loop integrand, thereby
reconstructing amplitudes from two-particle cuts that impose on-shell
constraints on internal lines. This requires an ansatz and typically algebra
at intermediate steps because many contributions contaminate the cuts. In
generalized unitarity, several propagators are placed on their mass-shell
simultaneously in order to isolate a small subset of integrals in a basis. At
one-loop, generalized unitarity has led to remarkably compact formulae for all
box, triangle and bubble integrals \cite{Britto:2004nc,Forde:2007mi} which are
now fully automated in several numerical implementations
\cite{Ellis:2007br,Berger:2008sj,Ossola:2007ax,
Mastrolia:2008jb,Giele:2008bc,Berger:2009zg,Badger:2010nx,Berger:2010zx,
Hirschi:2011pa} that are being applied to phenomenology at the LHC. 

Otherwise unattainable processes in massless QCD have been computed to at the
two-loop order with four external particles using the unitarity method and
other techniques including integration-by-parts identities
\cite{Bern:1997nh,Bern:2000dn,Glover:2001af,Bern:2002tk,Anastasiou:2000kg,
Anastasiou:2000ue,Anastasiou:2001sv}. In the last couple of years, new
promising methods for two-loop amplitudes such as integrand-level reduction by
multivariate polynomial division algorithms using Gr\"{o}bner bases and 
classification of on-shell solutions by primary decomposition based on
computational algebraic geometry have been demonstrated
\cite{Badger:2012dp,Mastrolia:2011pr,Zhang:2012ce,Feng:2012bm,
Mastrolia:2012an,Mastrolia:2012wf,Mastrolia:2012du,Kleiss:2012yv,
Huang:2013kh,Badger:2013gxa,Badger:2012dv}. In particular, very
recently the first results for five-gluon scattering at two loops in QCD were
obtained by Badger, Frellesvig and one of the present authors along these
lines \cite{Badger:2013gxa}. The same technique has also been applied to
maximal cuts of three-loop amplitudes \cite{Badger:2012dv}. 

In this paper, we will pursue amplitude computation at the level of integrated
expressions in the framework of maximal unitarity, proposed by Kosower and
Larsen \cite{Kosower:2011ty}. Being an intense version of the generalized
unitarity method, maximal unitarity is the natural continuation of the direct
extraction procedures for one-loop integrals. Maximal unitarity relies on a
unitarity comptatible basis of linearly independent integrals and the fact
that basis elements necessarily involve nontrivial tensorial numerators that
contaminate the cuts is a difficulty. However, steps towards developing a
general basis of planar two-loop integrals were recently taken in
\cite{Gluza:2010ws,Schabinger:2011dz}. The idea of maximal unitarity is to
isolate individual integrals in a basis by localizing their integrands onto
global poles and thereby cut as many propagators as possible. The loop
amplitude in question falls apart and becomes a product of its constituent
trees which upon integration along complex contours yield very compact final
expressions whose simplicity is diametrically opposite to that suggested by
the Feynman diagram approach.  Until now, integral bases for four-particle
amplitude contributions have been constructed by reduction to a set of master
integrals via integration-by-parts identities. Maximal unitarity has been
successfully applied in hepta-cuts of planar and non-planar double box
amplitude contributions with and without massive external legs in generic
gauge theories
\cite{Kosower:2011ty,CaronHuot:2012ab,Johansson:2012zv,Johansson:2012sf,
Larsen:2012sx,Sogaard:2013yga,Johansson:2013sda}. Previously octa-cuts and
hepta-cuts of two-loop amplitudes were addressed in maximally supersymmetric
$\mc N = 4$ super Yang-Mils theory \cite{Buchbinder:2005wp,Cachazo:2008vp}. 

The current status of maximal unitarity is that it has been demonstrated to
work for univariate residues. Indeed, hepta-cuts of two-loop amplitudes
generate algebraic curves \cite{CaronHuot:2012ab,Huang:2013kh} and the
unfrozen degree of freedom on the maximal cut is reflected by a left-over
one-dimensional contour integral that computes residues. It is common that
amplitudes at higher loops or below the leading singularity at the two-loop
level require several complex variables to parametrize a solution of a
unitarity cut. This necessarily leads to residues of higher-dimensional
differential forms. It is of obvious theoretical and practical interest to
extend the univariate unitarity method to general amplitude contributions. The
aim of this paper is therefore to develop a systematic way of determining
master integral coefficients in an integral basis using multivariate residues
in generalized unitarity cuts that define multidimensional algebraic
varieties.

The paper is organized in three main parts in the following way. In section 2
we introduce the mathematical prerequisites of multivariate residues and
computational algebraic geometry. Then, in section 3 we apply the formalism to
the planar triple box and derive a general formula for the master integral
coefficients.  Finally, in section 4 we obtain explicit results for the
three-loop triple box in any renormalizable gauge theory for independent
helicity configurations of four external gluons. We choose this primitive
amplitude because it was calculated short time ago, yet it is a highly
nontrivial test of our method. 

\section{Multivariate Residues}
Motivated by the discussion in the introduction we briefly review residues for
several complex variables. We focus on the computation of degenerate
multivariate residues and the global residue theorem. For mathematical
concepts, the references are the textbooks by Hartshorne \cite{MR0463157}, and
by Griffiths and Harris \cite{MR507725}.

We are working with $n$ complex variables, namely, $z\equiv (z_1, \dots,
z_n)$. First we study the local properties of a residue. Consider a residue at
$(z_1,\dots, z_n)=(\xi_1, \dots, \xi_n)\equiv\xi$. Let $U$ be the ball
$||z-\xi||<\epsilon$ and
assume that the functions $f_1(z), \dots, f_n(z)$ are holomorphic in a
neighborhood of the closure $\bar U$ of $U$, and have only one isolated
common zero, $\xi$ in $U$. Let $h(z)$ be a holomorphic function in
a neighborhood of $\bar U$. Then for the differential form,
\begin{equation}
  \label{eq:1}
  \omega=\frac{h(z) dz_1 \wedge \cdots \wedge dz_n}{f_1(z) \cdots f_n(z)}\;,
\end{equation}
the residue at $\xi$ regarding the function list $\{f_1,
  \dots, f_n\}$ is defined to be,
  \begin{equation}
    \label{local_residue}
    \Res{}_{\{f_1,
  \dots, f_n\},\xi}(\omega)=\bigg(\frac{1}{2\pi i}\bigg)^n\oint_{\Gamma}
\frac{h(z) dz_1 \wedge \cdots \wedge dz_n}{f_1(z) \cdots f_n(z)}\;,
  \end{equation}
where the contour $\Gamma$ is defined by the real $n$-cycle
$\Gamma=\{z: |f_i(z)|=\epsilon_i\}$ with the orientation specified by
the differential form $d(\arg f_1) \wedge \cdots \wedge d(\arg f_n)$.

The value of a local residue does not only depend on the location, but
also the ordering in the function list $\{f_1,
  \dots, f_n\}$. However, to simplify the notation, we may rewrite the
  left hand side of eqn. (\ref{local_residue}) as
  $\Res{}_{\xi}(\omega)$. The value of a residue is invariant under
  a non-singular change of complex coordinates. And we can rescale the
  function list, $f_i(z) \to \alpha_i(z) f_i(z)$, where $\alpha_i(z)$'s are
  holomorphic functions and $\alpha_i(\xi)\not=0$. Although the
  contour is changed, $\Res{}_{\{f_1,
  \dots, f_n\},\xi}(\omega)= \Res{}_{\{\alpha_1 f_1,
  \dots, \alpha_n f_n\},\xi}(\omega)$, from the fact that $\omega$ is closed and Stokes'
  theorem.

It is easy to prove that if locally $h(z)$ is generated by $\{f_1,
  \dots, f_n\}$, i.e., 
  \begin{equation}
  h(z)= a_1(z) f_1(z)+\cdots+a_n(z) f_n(z)\;,
  \end{equation}
  where the $a_i(z)$'s are holomorphic functions in a neighborhood of
  $\xi$, then 
  \begin{equation}
    \label{vanishing_th}
    \Res{}_{\{f_1,
  \dots, f_n\},\xi}(\omega)=0\;.
  \end{equation}
The fact comes from Stokes' theorem \cite{MR507725}.

We define the local residue to be {\it non-degenerate}, if the Jacobian of
$\{f_1,
  \dots, f_n\}$ at $\xi$ is nonzero,
  \begin{equation}
    \label{eq:2}
    J(\xi)\equiv \det\bigg(\frac{\partial f_i}{\partial z_j}(\xi)\bigg)\not=0\;.
  \end{equation}
In this case, the value of residue is simply \cite{MR507725},
\begin{equation}
    \label{non_degenerate_residue}
    \Res{}_{\{f_1,
  \dots, f_n\},\xi}(\omega)=\frac{h(\xi)}{J(\xi)}\;.
  \end{equation}
There is another type of residues which can be calculated
straightforwardly, {\it factorizable residues}. The definition is that  
each $f_i$ is a univariate polynomial, namely,
$f_i(z)=f_i(z_i)$. In this case, the contour in
eqn. (\ref{local_residue}) is factorized to the product of univariate
contours,
\begin{equation}
  \label{factorizable_residue}
   \Res{}_{\{f_1,
  \dots, f_n\},\xi}(\omega)=\bigg(\frac{1}{2\pi i}\bigg)^n \oint_{|f_1(z_1)|=\epsilon_1}
\frac{dz_1}{f_1(z_1)} \cdots
\oint_{|f_n(z_n)|=\epsilon_n} \frac{dz_n}{f_n(z_n)} h(z)\;.
\end{equation}
Then, by using the univariate residue formula $n$ times, we get the
value of the residue.

However, in general, a residue is neither non-degenerate nor
factorizable. Usually it is not convenient to use the definition
(\ref{local_residue}) to calculate a residue directly. Hence we need
new techniques for such computation.

\subsection{Calculation of (degenerate) local residues}
We restrict our computation to cases in which $f_1,\dots,f_n$
are polynomials in $(z_1,\dots,z_n)$. This is sufficient for all
unitarity calculations in quantum field theory. When all $f_i$'s are
polynomials, we can use the powerful tool of computational algebraic
geometry to obtain the residue. 

The key is to transform the polynomial list $\{f_1,\dots, f_n\}$ to a
new polynomial list $\{g_1, \dots, g_n\}$ such that the new residue
becomes factorable. So we recall the {\it transformation law},
\begin{thm}[Transformation law] Let $I=\langle f_1, \dots, f_n \rangle$ be the
 zero-dimensional ideal generated by $\{f_1,\dots,f_n\}$ and $J=\langle g_1,
\dots,
 g_n \rangle$ be a zero-dimensional ideal such that $J\subset I$. So
 $g_i = a_{ij} f_j$, where the $a_{ij}$'s are polynomials. Let $A$ be the
 matrix of $a_{ij}$'s, then for
 residues at $\xi$,
 \begin{equation}
   \label{transformation_law}
   \Res{}_{\{f_1,\dots,f_n\}, \xi}\bigg(\frac{h(z) dz_1 \wedge \cdots
     \wedge dz_n}{f_1(z) \cdots f_n(z)}\bigg) =  \Res{}_{\{g_1,\dots,
     g_n\}, \xi}   \bigg(\frac{h(z) dz_1 \wedge \cdots
     \wedge dz_n}{g_1(z) \cdots g_n(z)} \det A \bigg)\;.
 \end{equation}
\end{thm}
For the proof of this theorem, we refer to \cite{MR507725}. This theorem holds for both non-degenerate and degenerate residues.  

In practice, we use {\it Gr\"obner basis} method to find a list of
$\{g_1,\dots,g_n\}$ such that $g_i(z)=g_i(z_i)$. To obtain $g_i$, we
calculate the Gr\"obner basis of  $I=\langle f_1, \dots, f_n \rangle$
in {\it Lexicographic order} with the ordering $z_{i+1} \succ z_{i+2}
\succ \cdots z_n \succ z_1 \succ z_2 \succ \cdots \succ z_i$. Since this
lexicographic order would eliminate variables $z_1, \dots
,z_{i-1},z_{i+1}, \dots, z_n$, there would be a polynomial in the
basis which only depends on $z_i$. Then we call this polynomial
$g_i$. Repeating this process $n$ times, the desired $\{g_1,
\dots, g_n\}$ is obtained. After applying the
transformation law, the residue becomes factorable so it can
be calculated.

We demonstrate this algorithm by several examples,

\begin{example}
Let $n=2$, $\{f_1,f_2\}=\{z_1,(z_1+z_2)(z_1-z_2)\}$ and $h=z_2$. There
is a degenerate residue at $(0,0)$. Note that three factors $z_1$,
$z_1+z_2$ and $z_1-z_2$ vanish at the residue so the Jacobian is zero.
\begin{enumerate}
\item Calculate the Gr\"obner basis $G_1$, for $I=\langle z_1, (z_1 - z_2)
  (z_1 + z_2)\rangle$ in the lexicographic order $z_2 \succ z_1$. The result
  is that $G_1=\{z_1, z_2^2\}$. So we pick up the polynomial $g_1=z_1$. 
\item Calculate the Gr\"obner basis $G_2$, for $I$ in the
  lexicographic order $z_1 \succ z_2$, then $G_2=\{z_2^2, z_1\}$. Then
  we pick up  the polynomial $g_2=z_2^2$.
\item From the calculation of Gr\"obner basis,
  \begin{equation}
    \label{eq:5}
  \begin{pmatrix}
        g_1 \\
       g_2
 \end{pmatrix}=
  \begin{pmatrix}
        1 & 0\\
        z_1 & -1
 \end{pmatrix}  \begin{pmatrix}
        f_1 \\
       f_2
 \end{pmatrix}\;.
  \end{equation}
So $\det A=-1$. 
\end{enumerate}
The we calculate the residue by the transformation law,
\begin{equation}
  \label{eq:6}
  \Res{}_{\{f_1,f_2\}, (0,0)}\bigg(\frac{z_2 dz_1 \wedge
    dz_2}{z_1(z_1+z_2)(z_1-z_2)}\bigg) = \Res{}_{\{g_1, g_2\},
    (0,0)}\bigg(-\frac{z_2 dz_1 \wedge dz_2}{z_1 z_2^2}\bigg)=-1\;.
\end{equation}
Note that for this simple example, actually $G_1=G_2$. However, in general, we
need to calculate  
Gr\"obner basis $n$ times.
\end{example}

\begin{example}
Let $n=2$, $\{f_1,f_2\}=\{z_1^2,z_2-z_1\}$ and $h=z_2$. There
is a degenerate residue at $(0,0)$. Although there are only two
factors $z_1$, $z_2-z_1$ vanishing at the residue, the factor $z_1$
has the power $2$ so again the Jacobian is zero.

By Gr\"obner basis calculation,
 \begin{equation}
    \label{eq:5}
  \begin{pmatrix}
        z_1^2 \\
       z_2^2
 \end{pmatrix}=
  \begin{pmatrix}
        1 & 0\\
        1 & z_1+z_2
 \end{pmatrix}  \begin{pmatrix}
        f_1 \\
       f_2
 \end{pmatrix}\;.
  \end{equation}
\end{example}
Then, $\det(A)=z_1+z_2$, 
\begin{equation}
  \label{eq:6}
  \Res{}_{\{f_1,f_2\}, (0,0)}\bigg(\frac{z_2 dz_1 \wedge
    dz_2}{z_1^2 (z_2-z_1)}\bigg) = \Res{}_{\{g_1, g_2\},
    (0,0)}\bigg(\frac{z_2(z_1+z_2) dz_1 \wedge dz_2}{z_1^2 z_2^2}\bigg)=1\;.
\end{equation}

For our residue calculation of the three-loop triple box, we
frequently deal with the degenerate residues demonstrated in Example 1 and
2, as (1) more than $2$ factors vanish at the same point (2) one of the
vanishing factor has the power larger than $1$. These degenerate
residues are all calculated by Gr\"obner basis method. The computation
is automated by a program\footnote{We need the generator matrix of the
  Gr\"obner basis, which is not directly provided by Mathematica.} powered by the algebraic geometry software
Macaulay2 \cite{M2}.

Furthermore, when $n=2$, if three different factors $f_1$, $f_2$ and $f_3$
vanish at the
same point $\xi$, we have three different residues at $\xi$, namely,
\begin{equation}
  \Res{}_{\{f_1 f_2, f_3\},\xi}(\omega)\;,\quad 
  \Res{}_{\{f_2 f_3, f_1\},\xi}(\omega)\;,\quad 
  \Res{}_{\{f_3 f_1, f_2\},\xi}(\omega)\;.
\end{equation}
In general, the values of three residues are not the same. However, it
is clear that for a large number of classes of $f_i$'s, the sum of
three residues is zero. 

\begin{lemma}
  For three linear functions, $f_1(z_1,z_2)$, $f_2(z_1,z_2)$ and $f_3(z_1,z_2)$ vanishing
at
  $\xi$, such that $\langle f_1, f_2\rangle$, $\langle f_2, f_3\rangle$
  and $\langle f_3, f_1\rangle$ are all zero-dimensional ideals, let
  $h(z_1, z_2)$ be a holomorphic function in a neighborhood of $\xi$,
  and $\omega=h dz_1 \wedge dz_2/(f_1 f_2 f_3)$. Then,
\begin{equation}
  \Res{}_{\{f_1 f_2, f_3\}, \xi}(\omega)+ \Res{}_{\{f_2 f_3, f_1\},
    \xi} (\omega) + \Res{}_{\{f_3 f_1, f_2\}, \xi} (\omega) =0\;.
\label{triple_lemma}
\end{equation}
\end{lemma}
\begin{proof}
 We prove this identity by direct computation. Without loss of
 generality, we set $\xi=(0,0)$, $f_1=a_1 z_1+b_1 z_2$, $f_2=a_2
 z_1+b_2 z_2$ and $f_3=a_3 z_1+b_3 z_2$. By the Gr\"obner basis
 method, we have,
\begin{equation}
    \label{eq:5}
  \begin{pmatrix}
        z_1^2 \\
       z_2^2
 \end{pmatrix}=
  A_1\begin{pmatrix}
        f_1 f_2 \\
       f_3
 \end{pmatrix}\;, \quad 
\begin{pmatrix}
        z_1^2 \\
       z_2^2
 \end{pmatrix}=
  A_2\begin{pmatrix}
        f_2 f_3 \\
       f_1
 \end{pmatrix}\;, \quad 
\begin{pmatrix}
        z_1^2 \\
       z_2^2
 \end{pmatrix}=
  A_3\begin{pmatrix}
        f_3 f_1 \\
       f_2
 \end{pmatrix}\;,
  \end{equation}
where 
\begin{equation}
  \label{eq:10}
  \det A_1=\frac{\begin{vmatrix}
z_2 & z_1  \\     
 a_3 & b_3        
 \end{vmatrix}}{\begin{vmatrix}
        a_3 & b_3 \\
        a_1 & b_1
 \end{vmatrix}
\begin{vmatrix}
        a_3 & b_3 \\
        a_2 & b_2
 \end{vmatrix}
}\;,\quad
\det A_2=\frac{\begin{vmatrix}
z_2 & z_1  \\     
 a_1 & b_1        
 \end{vmatrix}}{\begin{vmatrix}
        a_1 & b_1 \\
        a_2 & b_2
 \end{vmatrix}
\begin{vmatrix}
        a_1 & b_1 \\
        a_3 & b_3
 \end{vmatrix}
}\;,\quad
\det A_3=\frac{\begin{vmatrix}
z_2 & z_1  \\     
 a_2 & b_2        
 \end{vmatrix}}{\begin{vmatrix}
        a_2 & b_2 \\
        a_1 & b_1
 \end{vmatrix}
\begin{vmatrix}
        a_2 & b_2 \\
        a_3 & b_3
 \end{vmatrix}
}\;.
\end{equation}
The condition that $\langle f_1, f_2\rangle$, $\langle f_2, f_3\rangle$
  and $\langle f_3, f_1\rangle$ are zero-dimensional ideals ensures
  that all the determinants in the denominators are nonzero. Then
  explicitly
  \begin{equation}
    \label{eq:11}
    \det A_1+\det A_2 +\det A_3=0\;,
  \end{equation}
by Laplace expansion of complementary minors. Then from the
transformation law, (\ref{triple_lemma}) is proven.
\end{proof}
Similarly, it is easy to check that the identity holds when one of the $f_i$
is a power of a
linear function or an irreducible quadratic polynomial. In our triple
box residue computation, we will use this lemma and its generalized
version to pick up independent residues.  

\subsection{Global Residue Theorem}
It is well known that for the univariate complex analysis, the sum of
residues of meromorphic differential form on $\mathbb {CP}^1$ is zero,
by the
residue theorem. We review the
multivariate version, the global residue theorem, in this
subsection. This theorem relates the residues in $\mathbb C^n$ and
residues at infinity, so it is crucial for the study of the global
structure of residues. 

\begin{thm}[Global residue residue, GRT] Let $M$ be a $n$-dimensional compact
complex manifold, $D_1,\dots,D_n$ be divisors of $M$. Assume that the
intersection $S=D_1 \cap D_2 \cdots \cap D_n$
only consists of discrete points. $\omega$ is a holomorphic
differential $n$-form on $M-D_1\cup D_2 \cdots \cup D_n$. Then
regarding the divisor list $\{D_1, D_2, \dots, D_n\}$,
\begin{equation}
  \label{GRT}
  \sum_{P\in S} \Res{}_P(\omega)=0\;.
\end{equation}
\end{thm}

A proof for this theorem can be found in \cite{MR507725}. Here are some explanations for this theorem: locally, the complex manifold $M$
is $\mathbb C^n$, so we can set up a local coordinate system
$(z_1,\dots, z_n)$. Also,
locally each divisor $D_i$ is an analytic hypersurface defined by
$f_i(z_1,\dots, z_n)=0$. Then the residue regarding $\{D_1, \dots,
D_n\}$ at $P$ is defined to be the local residue at $P$ regarding
$\{f_1, \dots, f_n\}$. From previous discussion, the value of the
residue is independent of the coordinate choice or local rescaling
of the $f_i$'s. So it is well-defined.

Although we may start our analysis from $\mathbb C^n$, the manifold
$\mathbb C^n$ is not compact so GRT does not apply directly. So we consider
the projective space $\mathbb{CP}^n$ with the homogenous coordinates $[w_0,
w_1, \dots, w_n]$. $\mathbb{CP}^n$ is covered by the open sets $U_i$,
where $U_i$ is defined to be the set 
$\{[w_0,w_1, \dots, w_n] | w_i\not=0\}$. $\mathbb C^n$ is embedded inside
$\mathbb{CP}^n$ as $U_0$,
\begin{equation}
  \label{eq:3}
  z_1=\frac{w_1}{w_0}\;, \quad \cdots \quad z_n=\frac{w_n}{w_0}\;.
\end{equation}
The points with the coordinate $w_0=0$ are called points at infinity,
which form the space $\mathbb{CP}^{n-1}$. $\mathbb{CP}^n$ is compact
and on which we can use GRT.

\begin{example}
  Consider $n=2$ and the differential form 
  \begin{equation}
    \label{eq:7}
    \omega=\frac{dz_1 \wedge dz_2}{z_1 z_2}\;.
  \end{equation}
We study this form on $\mathbb{CP}^2$. On the patch $U_1$,
$\omega=(dw_0 \wedge dw_1)/(w_0 w_1)$, while on the patch $U_2$,
$\omega=-(dw_0 \wedge dw_2)/(w_0 w_2)$. Hence, $\omega$ is defined on
$\mathbb{CP}^2$ excluding three irreducible hypersurfaces,
\begin{equation}
  \label{eq:9}
  w_0=0\;, \quad w_1=0\;, \quad w_2=0\;.
\end{equation}
To use GRT, we may define the two divisors, $D_1=\{w_0 w_1=0\}$ and
$D_2=\{w_2=0\}$. $\omega$ is defined on $\mathbb{CP}^2-D_1-D_2$. $D_1
\cap D_2=\{[1,0,0], [0,1,0]\}$, so there are two residues. GRT reads,
\begin{equation}
  \label{eq:12}
  \Res{}_{\{D_1,D_2\},[1,0,0]}(\omega)+\Res{}_{\{D_1,D_2\},[0,1,0]}(\omega)=(+1)+(-1)=0\;.
\end{equation}
However, it is also possible to define two divisors $D_1'=\{w_1=0\}$ and
$D_2'=\{w_0 w_2=0\}$. $D_1'
\cap D_2'
=\{[1,0,0], [0,0,1]\}$, so again there are two residues. Again GRT holds,
\begin{equation}
  \label{eq:13}
  \Res{}_{\{D_1',D_2'\},[1,0,0]}(\omega)+\Res{}_{\{D_1',D_2'\},[0,0,1]}(\omega)=(+1)+(-1)=0.
\end{equation}
So in summary, for $\omega$, apparently there is one residue at the origin,
two
residues at infinity. By using GRT twice for different choices of
divisors, we can see that values of the two residues at infinity are both the
opposite to that at origin. 
\end{example}
This simple example demonstrates GRT for multivariate
residues. In contrast to the univariate case, there are many different residue
identities for one differential form, by choosing different
divisors. In this paper, we will use GRT to study the residues at infinity.

\section{Maximal Cuts of Planar Triple Box Integrals}
We consider color-stripped four-gluon scattering amplitudes at three-loops in
gauge theories with $\text{SU}(N_c)$ symmetry group. The complete amplitude
receives contributions of a large number of integral topologies, for instance
the three-loop ladder with one or two crossed boxes and the planar tennis
court. Our focus is here on a primitive amplitude of modest complexity, the
planar triple box, although our method should generalize to other examples in
a straightforward manner.

The dimensionally regularized Feynman scalar integral for the four-point planar
triple box with massless kinematics on both internal and external lines reads
\begin{align}
\IP[1] \equiv 
\int_{\R^D}\frac{d^D\ell_1}{(2\pi)^D}
\int_{\R^D}\frac{d^D\ell_2}{(2\pi)^D}
\int_{\R^D}\frac{d^D\ell_3}{(2\pi)^D}
\prod_{i=1}^{10}
\frac{1}{D_i^2(\ell_1,\ell_2,\ell_3)}\;,
\label{TRIPLEBOXINT}
\end{align}
where the denominators or inverse propagators according to
fig.~\ref{TRIPLEBOXDIAGRAM} are given by
\begin{align}
D_1^2 = {} & \ell_1^2\;, &
D_2^2 = {} & \ell_2^2\;, &
D_3^2 = {} & \ell_3^2\;, &
D_4^2 = {} & (\ell_1+k_1)^2\;, \nn \\
D_5^2 = {} & (\ell_1-k_2)^2\;, &
D_6^2 = {} & (\ell_2+k_3)^2\;, &
D_7^2 = {} & (\ell_2-k_4)^2\;, &
D_8^2 = {} & (\ell_3+K_{12})^2\;, \nn \\
D_9^2 = {} & (\ell_1-\ell_3-k_2)^2\;, &
D_{10}^2 = {} & (\ell_3-\ell_2-k_3)^2\;.
\end{align}
All external momenta are outgoing and summed as 
$K_{i_1\cdots i_n} = k_{i_1}+\cdots+k_{i_n}$. Our definition of the triple box
integral is equivalent to that in \cite{Badger:2012dv} up to a linear
transformation of each of the three loop momenta,
\begin{align}
\tilde\ell_1 = \ell_1+k_1\;, \quad
\tilde\ell_2 = -\ell_2+k_4\;, \quad
\tilde\ell_3 = -\ell_3-K_{12}\;.
\end{align}
In general, \eqref{TRIPLEBOXINT} has a nontrivial polynomial numerator
function $\mc P(\ell_1,\ell_2,\ell_3)$ inserted and is in that situation
referred to as a tensor integral, although all Lorentz indices are properly
contracted. We specialize to four dimensions and therefore only reconstruct
the master integral coefficients to leading order in the dimensional regulator
$\epsilon$, leaving the general case to future work.
\begin{figure}[!h]
\bc
\includegraphics[scale=0.65]{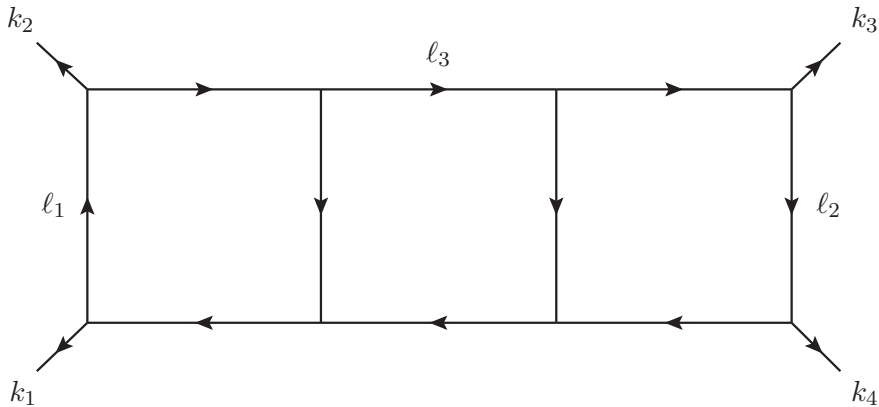}
\put(-315,-8){$k_1$}
\put(-316,133){$k_2$}
\put(0,133){$k_3$}
\put(0,-8){$k_4$}
\put(-303,63){$\ell_1$}
\put(-159,120){$\ell_3$}
\put(-13,63){$\ell_2$}
\caption{Momentum flow for the planar triple box.}
\label{TRIPLEBOXDIAGRAM}
\ec
\end{figure}

The ten inverse propagators generate a polynomial ideal 
$I = \avg{D_1,\dots,D_{10}}$ and the deca-cut equations, i.e. all internal
propagators are placed on their mass shell, define a two-dimensional algebraic
variety (or four-dimensional real surface) which is the zero locus of $I$,
usually denoted $\mc Z(I)$. In our notation, the solution set is 
\begin{align}
\mc S = \big\{(\ell_1,\ell_2,\ell_3)\in(\C^4)^{\otimes 3}\;|\;
D_i^2 = 0\,, \;\; i = 1,\dots,10\,\big\}\;.
\end{align}
By primary decomposition of the ideal it can be proven that the variety is 
reducible and the algebraic set can be decomposed uniquely into a union of 
fourteen components which are in one-to-one correspondence with the inequivalent
deca-cut solutions \cite{Badger:2012dv,Zhang:2012ce},
\begin{align}
\mc S = \bigcup_{i=1}^{14}\mc S_i\;, 
\quad \mc S_i\not{\!\!\subset}\,\mc S_j \;\; \mathrm{if}\;\; i\neq j\;.
\end{align}
The branches come in complex conjugate pairs, although, technically speaking,
identification by complex conjugation presumes reality of the external
momenta. We label the solutions $\mc S_1,\dots,\mc S_{14}$ such that $\mc S_i$
and $\mc S_{i+7}$ form a pair for $i = 1,\dots,7$. Note that $\mc S_i$ is not
a Riemann sphere because of the dimension.

This number of solutions is actually expected from general considerations of the 
$L$-loop ladder topology with chiral vertices. Indeed, such a diagram splits
into $2(L+1)$ three-point trees on the maximal cut with $(3L+1)$ propagators
placed on-shell. Let $N_L$ be the number of solutions at $L$ loops with $N_1 =
2$ and $N_2 = 6$. We can consider the distribution of holomorphically and
anti-holomorphically collinear vertices and easily realize that $N_L =
2(2^L-1)$. In particular, $N_3 = 14$, $N_4 = 30$, $N_5 = 62$ et cetera.

We proceed from here in three steps: first we solve the cut equations, then
we localize the scalar master integral onto the on-shell branches to expose
the composite leading singularities of the loop integrand, and finally we
derive unique projectors for all master integrals that are consistent with
integral reduction identities from parity-odd terms and total derivatives
which integrate to zero.

\subsection{Parametrization of On-Shell Solutions}
Uniformity of the dimension across all deca-cut branches implies that we can
parametrize the on-shell surfaces with two variables $(z_1,z_2)\in\C^2$. We 
exploit experience from previous calculations at one and two loops and choose 
convenient normalizations and basis elements such that the global poles of the 
integrand are directly exposed and as many of the decacut constraints as 
possible are linearized. Our parametrization of the three independent loop 
momenta with two-component Weyl spinors follows that of \cite{Badger:2012dv}
and reads
\begin{align}
\ell_1^\mu(\alpha_1,\dots,\alpha_4) = {} &
\alpha_1 k_1^\mu+\alpha_2 k_2^\mu+
\frac{\alpha_3}{2}\frac{\avg{23}}{\avg{13}}\expval{1^-}{\gamma^\mu}{2^-}+
\frac{\alpha_4}{2}\frac{\avg{13}}{\avg{23}}\expval{2^-}{\gamma^\mu}{1^-}\;,
\nn \\[1mm]
\ell_2^\mu(\beta_1,\dots,\beta_4)\, = {} &
\beta_1 k_3^\mu+\beta_2 k_4^\mu+
\frac{\beta_3}{2}
\frac{\avg{14}}{\avg{13}}
\expval{3^-}{\gamma^\mu}{4^-}+
\frac{\beta_4}{2}
\frac{\avg{13}}{\avg{14}}
\expval{4^-}{\gamma^\mu}{3^-}\;.
\nn \\[1mm]
\ell_3^\mu(\gamma_1,\dots,\gamma_4)\, = {} &
\gamma_1 k_2^\mu+\gamma_2 k_3^\mu+
\frac{\gamma_3}{2}
\frac{\avg{34}}{\avg{24}}
\expval{2^-}{\gamma^\mu}{3^-}+
\frac{\gamma_4}{2}
\frac{\avg{24}}{\avg{34}}
\expval{3^-}{\gamma^\mu}{2^-}\;,
\end{align}
All parameters are complex valued. In advance of calculations below we define
a frequently used ratio of Mandelstam invariants,
\begin{align}
\chi\equiv\frac{s_{14}}{s_{12}} = \frac{t}{s}\;.
\end{align}

In order to compensate for the change of variables from momenta to parameters 
in the triple box integral we in principle have to include the three Jacobian 
determinants
\begin{align}
J_\alpha = \det_{\mu,i}\pd{\ell_1^\mu}{\alpha_i} =
\frac{s_{12}^2}{4i}\;, \quad 
J_\beta = \det_{\mu,i}\pd{\ell_2^\mu}{\beta_i} =
\frac{s_{12}^2}{4i}\;, \quad
J_\gamma = \det_{\mu,i}\pd{\ell_3^\mu}{\gamma_i} =
\frac{s_{14}^2}{4i}\;.
\label{PARAMETERJACOBIANS}
\end{align}
They are however constant as an immediate consequence of linearity of the 
parametrization of the loop momenta and can thus be disregarded in the 
augmented deca-cut in what proceeds.

Let us write out the deca-cut equations using the parametrizations of
$\ell_1$, $\ell_2$ and $\ell_3$. It is a straightforward task and we
immediately derive all equations involving only one loop momentum and one
extenral leg,
\begin{align}
D_1^2 = {} & s_{12}\big(\alpha_1\alpha_2-\alpha_3\alpha_4\big) = 0\;, &
D_2^2 = {} & s_{12}\big(\beta_1\beta_2-\beta_3\beta_4\big) = 0\;, \nn \\[1mm]
D_3^2 = {} & s_{14}\big(\gamma_1\gamma_2-\gamma_3\gamma_4\big) = 0\;, &
D_4^2 = {} & s_{12}\big((\alpha_1+1)\alpha_2-\alpha_3\alpha_4\big) = 0\;, \nn \\[1mm]
D_5^2 = {} & s_{12}\big(\alpha_1(\alpha_2-1)-\alpha_3\alpha_4\big) = 0\;, &
D_6^2 = {} & s_{12}\big((\beta_1+1)\beta_2-\beta_3\beta_4\big) = 0\;, \nn \\[1mm]
D_7^2 = {} & s_{12}\big(\beta_1(\beta_2-1)-\beta_3\beta_4\big) = 0\;,
\end{align}
Moreover, we have
\begin{align}
D_8^2 = {} & s_{12}\big(1+\gamma_1-\gamma_2-\gamma_3+
(1+\chi)\gamma_4+\chi(\gamma_1\gamma_2-\gamma_3\gamma_4)\big) = 0\;.
\end{align}
For generic external kinematics, we readily get the trivial constraints 
$\alpha_1 = \alpha_2 = 0$ and $\beta_1 = \beta_2 = 0$ together with 
$\alpha_3\alpha_4 = 0$ and $\beta_3\beta_4 = 0$. It then remains to consider
the inverse propagators on the two middle rungs of the triple box. We omit for 
brevity the full expansions in parameter space, and display only equations 
simplified on the constraints for $\ell_1$ and $\ell_2$ above,
\begin{align}
D_9^2 = {} &
s_{12}\left\{-\alpha_3\left[\alpha_4+\chi\left(\gamma_2
+\frac{\gamma_3}{1+\chi}\right)\right]
+\chi(\gamma_2(1+\gamma_1)
-\gamma_3\gamma_4)+(1+\chi)\alpha_4(\gamma_2-\gamma_4)\right\}\;, \nn \\[1mm]
D_{10}^2 = {} &
s_{12}\left\{
-\left(\gamma_4+\frac{\beta_3}{1+\chi}\right)[(1+\chi)\beta_4+\chi\gamma_3]
-\gamma_1(\beta_4+(1-\gamma_2-\beta_3+\beta_4)\chi)\right\}\;.
\end{align}

The complete set of local solutions to these equations were reported in
\cite{Badger:2012dv}. We quote the result for $\mc S_1,\dots,\mc S_7$ below
for completeness. Also, the complex conjugates denoted with primes can be
constructed easily using the following relations valid for each pair,
\begin{align}
\alpha_1' = {} & 0\;, & \alpha_2' = {} & 0\;, &
\alpha_3' = {} & -\frac{1+\chi}{\chi}\alpha_4\;, &
\alpha_4' = {} & -\frac{\chi}{1+\chi}\alpha_3\;, \nn \\
\beta_1' = {} & 0\;, & \beta_2' = {} & 0\;, & 
\beta_3' = {} & -\frac{1+\chi}{\chi}\beta_4\;, &
\beta_4' = {} & -\frac{\chi}{1+\chi}\beta_3\;, \nn \\
\gamma_1' = {} & \gamma_1\;, & \gamma_2' = {} & \gamma_2\;, & 
\gamma_3' = {} & -(1+\chi)\gamma_4\;, &
\gamma_4' = {} & -\frac{1}{1+\chi}\gamma_3\;. \nn
\end{align}

On the maximal cut, all eight vertices have three massless legs attached. We
point out that the fourteen on-shell branches are in one-to-one correspondence
with the valid kinematical configurations of three-point trees in the
maximally cut planar triple box such that no external legs are neither
holomorphically nor antiholomorphically collinear for generic momenta. 
All diagrams not related by complex conjugation are included in
Appendix~\ref{TREESMP}.
\begin{table}[!h]
\begin{align}
\begin{array}{|c|c|c|c|c|c|c|c|c|}
\hline
\;&\; \alpha_1 \;&\; \alpha_2 \;&\; \alpha_3 \;&\; \alpha_4 
\;&\; \beta_1 \;&\; \beta_2 \;&\; \beta_3 \;&\; \beta_4 \\  
\hline
\;\mc S_1 
\;&\,\; 0 \;\,&\,\; 0 \;\,&\,\; 1-\frac{1}{\chi}\frac{1+z_1}{z_2} 
\;\,&\,\; 0 \;\,&\,\; 0 \;\,&\,\; 0 \;\,&\,\; 1+z_2 \;\,&\,\; 0 \\
\;\mc S_2 
\;&\,\; 0 \;\,&\,\; 0 \;\,&\,\; \left(1+\frac{1}{\chi}\right)(1+z_1) \;\,&\,\; 0
\;\,&\,\; 0 \;\,&\,\; 0 \;\,&\,\; 0 \;\,&\,\; -1-\frac{1}{1+\chi}\frac{1}{z_1} \\
\;\mc S_3
\;&\,\; 0 \;\,&\,\; 0 \;\,&\,\; z_2 \;\,&\,\; 0
\;\,&\,\; 0 \;\,&\,\; 0 \;\,&\,\; 0 \;\,&\,\; -1-z_1 \\
\;\mc S_4
\;&\,\; 0 \;\,&\,\; 0 \;\,&\,\; \left(1+\frac{1}{\chi}\right)(1+z_1) 
\;\,&\,\; 0 \;\,&\,\; 0 \;\,&\,\; 0 \;\,&\,\; 0 \;\,&\,\; z_2 \\
\;\mc S_5
\;&\,\; 0 \;\,&\,\; 0 \;\,&\,\; z_2 \;\,&\,\; 0
\;\,&\,\; 0 \;\,&\,\; 0 \;\,&\,\; 1 \;\,&\,\; 0 \\
\;\mc S_6
\;&\,\; 0 \;\,&\,\; 0 \;\,&\,\; 1 \;\,&\,\; 0
\;\,&\,\; 0 \;\,&\,\; 0 \;\,&\,\; z_2 \;\,&\,\; 0 \\
\;\mc S_7
\;&\,\; 0 \;\,&\,\; 0 \;\,&\,\; z_1 \;\,&\,\; 0
\;\,&\,\; 0 \;\,&\,\; 0 \;\,&\,\; z_2 \;\,&\,\; 0 \\
\hline
\end{array}
\nn
\end{align}
\label{TRIPLEBOXSOLUTIONS}
\caption{Values of the coefficients in the outermost loop momenta $\ell_1$ and 
$\ell_2$ written in terms of the two unfrozen parameters $(z_1,z_2)\in\C^2$.}
\end{table}
\begin{table}[!h]
\begin{align}
\begin{array}{|c|c|c|c|c|}
\hline
\,&\, \gamma_1 \,&\, \gamma_2 \,&\, \gamma_3 \,&\, \gamma_4 \\
\hline
\,\mc S_1 
\,&\, 
\frac{1}{\chi}\left(1+\frac{1}{z_1}\right)\left(1+\frac{1}{z_2}\right) 
\,&\, 
-\frac{1}{\chi}\left(1+\frac{1}{z_1}\right)+\frac{z_2}{z_1}
\,&\, 
\frac{1+\chi}{\chi}\left(1+\frac{1}{z_1}\right) 
\,&\, \beta_4(\mc S_1) \\
\,\mc S_2 
\,&\, 
\frac{1}{\chi}(1+z_2)(1+(1+\chi)z_1) 
\,&\,
\frac{1}{\chi}\left(1+\frac{1}{z_1}\right)z_2
\,&\, 
-\left(\frac{1+\chi}{\chi}+\frac{1}{\chi}\frac{1}{z_1}\right)z_2
\,&\, 
-\frac{1}{\chi}(1+z_1)(1+z_2) \\ 
\,\mc S_3
\,&\, \frac{1}{\chi}\left(1+\frac{1}{z_1}\right) \,&\, 0 \,&\, 0 
\,&\, -\frac{1}{\chi}\left(1+\frac{1}{1+\chi}\frac{1}{z_1}\right) \\
\,\mc S_4
\,&\, 0 \,&\, -\frac{1}{\chi}\left(1+\frac{1}{z_1}\right) 
\,&\, \frac{1}{\chi}\frac{1}{z_1}+\frac{1+\chi}{\chi}
\,&\, 0 \\ 
\,\mc S_5
\,&\, z_1 \,&\, 0 \,&\, 0 \,&\, -\frac{1}{1+\chi}(1+z_1) \\
\,\mc S_6
\,&\, 0 \,&\, z_1 \,&\, 0 \,&\, -\frac{1}{1+\chi}(1-z_1) \\
\,\mc S_7
\,&\, 0 \,&\, 0 \,&\, 0 \,&\, -\frac{1}{1+\chi} \\
\hline
\end{array}
\nn
\end{align}
\label{TRIPLEBOXSOLUTIONSMIDDLE}
\caption{Solutions for loop momentum $\ell_3$.
$\beta_4(\mc S_1)\equiv
-\frac{1}{\chi(1+\chi)}\left(1+\frac{1}{z_1}\right)\left(1+\frac{1}{z_2}\right)
+\frac{1}{1+\chi}\frac{1+z_2}{z_1}$.}
\end{table}

\subsection{Composite Leading Singularities}
Loop momenta that solve simultaneous on-shell contraints are complex valued
for general external kinematics. As a direct conseqeuence, the traditional
unitarity cut procedure involving several Dirac delta functions that enforce
on-shell conditions would yield a trivial equation reflecting that their
support is zero and this is not sufficient to extract the integral
coefficients. This leads to an obstacle already for quadruple cuts at one
loop in the work by Britto, Cachazo and Feng \cite{Britto:2004nc}, later
clarified by Kosower and Larsen \cite{Kosower:2011ty}.

Instead of a integration region that comprises real Minkowski space, loop
integrals should really be reinterpreted in complex momenta,
$\ell_i^\mu\in\C^4$, evaluated along the real slices $\text{Im}\,\ell_i^\mu =
0$. Generalized unitarity cuts that are also appropriate to complex solutions
may then be realized by replacement of contours by a higher-dimensional
surface that is topologically equivalent to a torus encircling a global poles
of the loop integrand. Recall the transformation of localization property of
multidimensional contour integrals, which we immediately recognize as the
natural generalization of the Dirac delta function to several complex
variables once we for $\xi\in\C^n$ apply the substitution
\begin{align}
\int dz_1\dots dz_n h(z_i)\prod_{j=1}^n\delta(z_j-\xi_j)\longrightarrow
\frac{1}{(2\pi i)^n}\oint_{\Gamma_\epsilon(\xi)}\!
dz_1\wedge\cdots \wedge dz_n\frac{h(z_i)}{\prod_{j=1}^n(z_j-\xi_j)}\;.
\end{align}
Notice that, crucially, analyticity of the integrand is maintained in virtue
of absence of absolute values in the Jacobian. As a consequence we are able to
carry out successive contour integations. 

Let us apply this strategy to compute the localization of the triple box
scalar integral onto the fourteen on-shell branches and begin to explore the
composite leading singularities of the integrand. Each integral must be
treated individually. However, due to the appearance of complex conjugate
pairs of solutions to the deca-cut equations, we expect at most seven distinct
remaining two-dimensional complex integrals. The general object of our
interest is
\begin{align}
\IP[1]_{\mc S_i} \equiv \frac{J_\alpha J_\beta J_\gamma}{(2\pi i)^{12}}
\oint_{T_{\alpha,i}}\!\frac{d^4\alpha}{(2\pi)^4}
\oint_{T_{\beta,i}}\!\frac{d^4\beta}{(2\pi)^4}
\oint_{T_{\gamma,i}}\!\frac{d^4\gamma}{(2\pi)^4}
\prod_{k=1}^{10}\frac{1}{D_k^2(\alpha,\beta,\gamma)}\;,
\end{align}
where $T_{\alpha,i}\times T_{\beta,i}\times T_{\gamma,i}$ is the 
twelwe-dimensional leading singularity cycle associated with solution 
$\mc S_i$ and the prefactors are given in \eqref{PARAMETERJACOBIANS}.
This multivariate residue is nondegenerate, hence performing ten of the
integrals by means of the localization property for multidimensional contour
integrals \eqref{non_degenerate_residue} leaves us with a two-dimensional
contour integral and the deca-cut Jacobian $J_i$ defined via
\begin{align}
\IP[1]_{\mc S_i} = 
\frac{J_\alpha J_\beta J_\gamma}{(2\pi i)^2}
\oint\!\frac{dz_1\wedge dz_2}{J_i(z_1,z_2)}\;.
\end{align}

To set the stage we derive the expression for $J_1$ in full detail. In this 
case we keep integrals over $\alpha_3$ and $\beta_3$. Other choices yield 
equivalent end results as long as the coordinate transformation to $z_1$ and 
$z_2$ is nonsingular. Indeed, this change of variables introduces of course yet
another Jacobian which is trivial to compute. Anyway, the Jacobian that arise
upon cutting the ten propagators on-shell is block diagonal and decomposes into 
three determinants. In particular, we denote the Jacobian from the six deca-cut
equations involving only either loop momentum $\ell_1$ or $\ell_2$ by $J_A$, 
\begin{align}
\hspace*{-1mm}
J_A^{-1} = \frac{1}{(2\pi i)^6s_{12}^6}
\oint_{C^3_\epsilon(0)}d^3\alpha
\frac{1}{\alpha_1\alpha_2-\alpha_3\alpha_4}
\frac{1}{(\alpha_1+1)\alpha_2-\alpha_3\alpha_4}
\frac{1}{\alpha_1(\alpha_2-1)-\alpha_3\alpha_4} 
\nn \hspace*{1cm} \\[2mm] \times
\oint_{C^3_\epsilon(0)}d^3\beta
\frac{1}{\beta_1\beta_2-\beta_3\beta_4}
\frac{1}{(\beta_1+1)\beta_2-\beta_3\beta_4}
\frac{1}{\beta_1(\beta_2-1)-\beta_3\beta_4}\;,
\end{align}
in a notation that should be self-explanatory. It is then easy to get
\begin{align}
J_A = s_{12}^6\det\left(
\begin{array}{ccc}
\alpha_2 & \alpha_1 & -\alpha_3 \\
\alpha_2 & \alpha_1+1 & -\alpha_3 \\
\alpha_2-1 & \alpha_1 & -\alpha_3
\end{array}
\right) 
\det\left(
\begin{array}{ccc}
\beta_2 & \beta_1 & -\beta_3 \\
\beta_2 & \beta_1+1 & -\beta_3 \\
\beta_2-1 & \beta_1 & -\beta_3
\end{array}
\right) = 
s_{12}^6\alpha_3\beta_3\;.
\end{align}
The next step is to integrate over the four constraints involving $\ell_3$ and
evaluate the determinant which we call $J_B$. Prior to insertion of the
explicit cut solution, this expression expands to a slightly complicated form
which we do not include here,
\begin{align}
J_B = 
\chi^4s_{12}^4\det\left(
\begin{array}{cccc}
\gamma_2 & \gamma_1 & -\gamma_4 & -\gamma_3 \\
\gamma_2+\frac{1}{\chi} & \gamma_1-\frac{1}{\chi} & -\gamma_4-\frac{1}{\chi} &
-\gamma_3+1+\frac{1}{\chi} \\
\gamma_2 & 1-\alpha_3+\gamma_1 & -\gamma_4-\frac{\alpha_3}{1+\chi} & -\gamma_3 \\
-1+\gamma_2+\beta_3 & \gamma_1 & -\gamma_4-\frac{\beta_3}{1+\chi} & -\gamma_3
\end{array}
\right)\;.
\end{align}
Finally we apply the coordinate transformation $(\alpha_3,\beta_3)\to(z_1,z_2)$ 
and obtain the integral
\begin{align}
\IP[1]_{\mc S_1} = +\frac{1}{(2\pi i)^2\chi^2s_{12}^{10}}\oint
\frac{dz_1\wedge dz_2}{(1+z_1)(1+z_2)(1+z_1-\chi z_2)}\;.
\end{align}

We repeated the calculation for the remaining thirteen deca-cut solutions and
verified that complex conjugate pairs indeed have the same Jacobian up to an
ambiguous overall sign due to antisymmetry of determinants, i.e. 
\begin{align}
\IP[1]_{\mc S_{i+7}} = \IP[1]_{\mc S_i}
\end{align}
for $i = 1,\dots,7$. In summary, we found the following relatively simple 
integrals
\begin{gather}
\IP[1]_{\mc S_2} = -\frac{1}{(2\pi i)^2\chi^2s_{12}^{10}}\oint
\frac{dz_1\wedge dz_2}{(1+z_1)(1+z_2)(1+(1+\chi)z_1)z_2}\;, \\[3mm]
\IP[1]_{\mc S_3} = -\frac{1}{(2\pi i)^2\chi^2s_{12}^{10}}\oint
\frac{dz_1\wedge dz_2}{(1+z_1)z_2(1+z_1[1+\chi(1-z_2)])}\;, \\[3mm]
\IP[1]_{\mc S_4} = -\frac{1}{(2\pi i)^2\chi^2s_{12}^{10}}\oint
\frac{dz_1\wedge dz_2}{(1+z_1)z_2(1+(1+\chi)z_1(1+z_2))}\;, \\[3mm]
\IP[1]_{\mc S_5} = -\frac{1}{(2\pi i)^2\chi^3s_{12}^{10}}\oint
\frac{dz_1\wedge dz_2}{z_1z_2(1+z_1-z_2)}\;, \\[3mm]
\IP[1]_{\mc S_6} = +\frac{1}{(2\pi i)^2\chi^3s_{12}^{10}}\oint
\frac{dz_1\wedge dz_2}{z_1z_2(1-z_1-z_2)}\;, \\[3mm]
\IP[1]_{\mc S_7} = -\frac{1}{(2\pi i)^2\chi^3s_{12}^{10}}\oint
\frac{dz_1\wedge dz_2}{z_1z_2(1-z_1)(1-z_2)}\;,
\end{gather}
We emphasize that the number of polynomials in the denominators and hence
residues to consider differs from branch to branch, which is not the case in
the two-loop double box, reflecting a more complicated global topological
picture of the underlying algebraic variatey. In general, these contour
integrals have a nontrivial polynomial numerator which is evaluated on the
particular branch. Although the Jacobian computations are not affected by the
presence of a numerator function, the analytic structure of the integrands may
change significantly.

Notice that the cut integrals in solution $\mc S_7$ and its complex conjugate
are particularly simple because they separate onto univariate contours with
factorizable residues which can be computed iteratively as in 
\eqref{factorizable_residue}.

\subsection{Global Poles and Augmentation of Residues}
We expand the planar triple box contribution to the four-gluon three-loop
amplitude onto an integral basis of three master integrals, say,
\begin{align}
A_{4,(3)}^{\square\!\square\!\square}
(1^{\lambda_1},2^{\lambda_2},3^{\lambda_3},4^{\lambda_4}) = {} &
C_1(1^{\lambda_1},2^{\lambda_2},3^{\lambda_3},4^{\lambda_4})\IP[1]+ \nn \\[1mm] {} &
C_2(1^{\lambda_1},2^{\lambda_2},3^{\lambda_3},4^{\lambda_4})
\IP[(\tilde\ell_1+k_4)^2]+ \nn \\[2mm] {} &
C_3(1^{\lambda_1},2^{\lambda_2},3^{\lambda_3},4^{\lambda_4})
\IP[(\tilde\ell_3-k_4)^2]+
\cdots
\label{TRIPLEBOXMIEQ}
\end{align}
where the ellipses cover terms with less than ten propagators which cannot be
probed by maximal unitarity. It of course possible to construct other 
equivalent integral bases for the ladder topology via integration-by-parts 
identities, however this choice allows to directly compare with the results of 
Badger, Frellesvig and Zhang. The analytic expression of the scalar master
integral $\IP[1]$ is known in dimensional regularization
\cite{Smirnov:2003vi}.

Our problem is now reduced to extract the three master integral coefficients to
$\mc O(\epsilon^0)$ from the augmented deca-cut of \eqref{TRIPLEBOXMIEQ}. It is
customary to strip off universal prefactors and define dimensionless 
coefficients $\hat C_i$ by
\begin{align}
C_i(1^{\lambda_1},2^{\lambda_2},3^{\lambda_3},4^{\lambda_4}) =
s_{12}^3s_{14}A_4^\tree(1^{\lambda_1},2^{\lambda_2},3^{\lambda_3},4^{\lambda_4})
\hat C_i(1^{\lambda_1},2^{\lambda_2},3^{\lambda_3},4^{\lambda_4})\;,
\end{align}
where the four-gluon tree amplitude is given by the Parke-Taylor formula, for
instance for helicities $--++$,
\begin{align}
A_4^\tree(1^-,2^-,3^+,4^+) =
i\frac{\avg{12}^4}{\avg{12}\avg{23}\avg{34}\avg{41}}\;.
\end{align}
Notice that we form now on leave helicity labels on coefficients implicit. Then
we can write down the augmented deca-cut of \eqref{TRIPLEBOXMIEQ}
schematically. The triple box amplitude contribution falls apart into a
product of eight on-shell three-particle amplitudes such that
\begin{align}
\frac{1}{(2\pi i)^2}
\sum_{i=1}^{14}\oint_{\Gamma_i}\frac{dz_1\wedge dz_2}{J_i(z_1,z_2)}
\sum_{\substack{\text{helicities}\\\text{particles}}}
\prod_{j=1}^8 A_{(j)}^\tree(z_1,z_2)\big|_{\mc S_i} = \hspace*{6cm} \nn 
\\[-2mm]
C_1\IP[1]_{\cut}+
C_2\IP[(\tilde\ell_1+k_4)]_{\cut}+
C_3\IP[(\tilde\ell_3-k_4)]_{\cut}\;,
\label{AUGDECA}
\end{align}
where the tree-level data evaluated in the deca-cut kinematics defines fourteen
multivariate Laurent polynomials,
\begin{align}
\sum_{\substack{\text{helicities}\\\text{particles}}}
\prod_{j=1}^8 A_{(j)}^\tree(z_1,z_2)\big|_{\mc S_i} = 
\sum_{j,k} d_{i;jk}z_1^jz_2^k\;.
\end{align}
All coefficients with $|i|$ or $|j| > 4$ vanish in all gauge theories
even though higher order terms are allowed in renormalization
\cite{Badger:2012dv}. It follows from the explicit solutions that the Laurent
series for $\mc S_1$ and its complex conjugate $\mc S_8$ contain terms with
negative powers of $z_1$ and $z_2$ in the numerators. Furthermore, parameters
in $\mc S_2,\dots,\mc S_4$ and their complex conjugates contain negative powers of
$z_1$. These will also give rise to global poles. 

We therefore accordingly construct $14$ residue loci from all combinations of
two vanishing polynomial factors in the deca-cut Jacobians and the possible
extra denominators $z_1$, $z_2$. Here we list all residues for finite values
of $(z_1, z_2)$. The residues are labeled by the number of branch in Roman
numerals, and its order in this branch. Multivariate residues also
depend on the choice and order of the denominators. So in 
$(...)$ we give the location of the residue, and in $\{\ldots\}$ we
list the denominators vanishing on the residue. There are also
residues at infinity. However, we prove that for the maximal cut of the
triple box, all residues at infinity are related to poles at finite loci by
the {\it Global residue theorem (GRT)}. So we do not list residues at
infinity and we give the proof in Appendix \ref{poles_at_infinity}.
\begin{itemize}
\item Branch $\mc S_1$: $8$ residues for finite $z_1$, $z_2$:

\begin{tabular}{ll}
I1\; $(-1,-1)\,,\; \{1+z_1,1+z_2\}$ & I2\; $(-1,0)\,,\;
\{(1+z_1)z_2,1+z_1-\chi z_2\}$\\
I3\; $(-1,0)\,,\; \{z_2(1+z_1-\chi z_2),1+z_1\}$ & I4\; $(-1,0)\,,\;
\{(1+z_1)(1+z_1-\chi z_2),z_2\}$\\
I5\; $(0,-1)\,,\; \{z_1,1+z_2\}$ & I6\; $(0,0)\,,\; \{z_1,z_2\}$\\
I7\; $(0,1/\chi)\,,\; \{z_1,1+z_1-\chi z_2\}$ & I8\; $(-1-\chi,-1)\,,\; \{1+z_2,1+z_1-\chi z_2\}$
\end{tabular}

\item Branch $\mc S_2$: $6$ residues for finite $z_1$, $z_2$:

\begin{tabular}{ll}
II1\; $(-1,-1)\,,\; \{1+z_1,1+z_2\}$ & II2\; $(-1,0)\,,\;
\{1+z_1,z_2\}$\\
II3\;  $(0,-1)\,,\; \{z_1,1+z_2\}$ & II4\; $(0,0)\,,\; \{z_1,z_2\}$\\
II5\;  $(-1/(1+\chi),-1)\,,\; \{1+(1+\chi)z_1,1+z_2\}$ & II6\; $(-1/(1+\chi),0)\,,\; \{1+(1+\chi)z_1,z_2\}$
\end{tabular}

\item Branch $\mc S_3$: $4$ residues for finite $z_1$, $z_2$:

\begin{tabular}{ll}
III1\; $(-1,0)\,,\; \{1+z_1,z_2\}$ & III2\; $(-1,1)\,,\;
\{1+z_1,1+z_1(1+\chi(1-z_2))\}$\\
III3\; $(0,0)\,,\; \{z_2,z_1\}$ & III4\; $(-1/(1+\chi),0)\,,\; \{z_2,1+z_1(1+\chi(1-z_2))\}$
\end{tabular}

\item Branch $\mc S_4$: $4$ residues for finite $z_1$, $z_2$:

\begin{tabular}{ll}
IV1\; $(-1,0)\,,\; \{1+z_1,z_2\}$ & IV2\; $(-1,-\chi/(1+\chi))\,,\; \{1+z_1,1+(1+\chi)z_1(1+z_2)\}$\\
IV3\; $(0,0)\,,\; \{z_1,z_2\}$ & IV4\; $(-(1/(1+\chi),0)\,,\; \{z_2,1+(1+\chi)z_1(1+z_2)\}$
\end{tabular}

\item Branch $\mc S_5$: $4$ residues for finite $z_1$, $z_2$:

\begin{tabular}{ll}
V1\; $(-1,0)\,,\; \{1+z_1-z_2,z_2\}$ & V2\; $(0,0))\,,\; \{z_1,z_2\}$
\\
V3\; $(0,1)\,,\; \{z_1,1+z_1-z_2\}$
\end{tabular}

\item Branch $\mc S_6$: $3$ residues for finite $z_1$, $z_2$:

\begin{tabular}{ll}
VI1\; $(0,0)\,,\; \{z_1,z_2\}$ & VI2\; $(0,1)\,,\; \{z_1,z_1+z_2-1\}$
\\
VI3\; $(1,0)\,,\; \{z_2,z_1+z_2-1\}$
\end{tabular}

\item Branch $\mc S_7$: $4$ residues for finite $z_1$, $z_2$:

\begin{tabular}{ll}
VII1\; $(0,0)\,,\; \{z_1,z_2\}$ & VII2\; $(0,1)\,,\; \{z_1,z_2-1\}$
\\
VII3\; $(1,0)\,,\; \{z_1-1,z_2\}$ & VII4\;$(1,1)\,,\; \{z_1-1,z_2-1\}$
\end{tabular}

\clearpage
\item Branch $\mc S_8$: $8$ residues for finite $z_1$, $z_2$:

\begin{tabular}{ll}
VIII1\; $(-1,-1)\,,\; \{1+z_1,1+z_2\}$ & VIII2\; $(-1,0)\,,\;
\{(1+z_1)z_2,1+z_1-\chi z_2\}$\\
VIII3\; $(-1,0)\,,\; \{z_2(1+z_1-\chi z_2),1+z_1\}$ & VIII4\; $(-1,0)\,,\;
\{(1+z_1)(1+z_1-\chi z_2),z_2\}$\\
VIII5\; $(0,-1)\,,\; \{z_1,1+z_2\}$ & VIII6\; $(0,0)\,,\; \{z_1,z_2\}$\\
VIII7\; $(0,1/\chi)\,,\; \{z_1,1+z_1-\chi z_2\}$ & VIII8\; $(-1-\chi,-1)\,,\; \{1+z_2,1+z_1-\chi z_2\}$
\end{tabular}

\item Branch $\mc S_9$: $6$ residues for finite $z_1$, $z_2$:

\begin{tabular}{ll}
IX1\; $(-1,-1)\,,\; \{1+z_1,1+z_2\}$ & IX2\; $(-1,0)\,,\;
\{1+z_1,z_2\}$\\
IX3\;  $(0,-1)\,,\; \{z_1,1+z_2\}$ & IX4\; $(0,0)\,,\; \{z_1,z_2\}$\\
IX5\;  $(-1/(1+\chi),-1)\,,\; \{1+(1+\chi)z_1,1+z_2\}$ & IX6\; $(-1/(1+\chi),0)\,,\; \{1+(1+\chi)z_1,z_2\}$
\end{tabular}

\item Branch $\mc S_{10}$: $4$ residues for finite $z_1$, $z_2$:

\begin{tabular}{ll}
X1\; $(-1,0)\,,\; \{1+z_1,z_2\}$ & X2\; $(-1,1)\,,\;
\{1+z_1,1+z_1(1+\chi(1-z_2))\}$\\
X3\; $(0,0)\,,\; \{z_2,z_1\}$ & X4\; $(-1/(1+\chi),0)\,,\; \{z_2,1+z_1(1+\chi(1-z_2))\}$
\end{tabular}

\item Branch $\mc S_{11}$: $4$ residues for finite $z_1$, $z_2$:

\begin{tabular}{ll}
XI1\; $(-1,0)\,,\; \{1+z_1,z_2\}$ & XI2\; $(-1,-\chi/(1+\chi))\,,\;
\{1+z_1,1+(1+\chi)z_1(1+z_2)\}$\\
XI3\; $(0,0)\,,\; \{z_1,z_2\}$ & XI4\; $(-(1/(1+\chi),0)\,,\; \{z_2,1+(1+\chi)z_1(1+z_2)\}$
\end{tabular}

\item Branch $\mc S_{12}$: $4$ residues for finite $z_1$, $z_2$:

\begin{tabular}{ll}
XII1\; $(-1,0)\,,\; \{1+z_1-z_2,z_2\}$ & XII2\; $(0,0))\,,\; \{z_1,z_2\}$
\\
XII3\; $(0,1)\,,\; \{z_1,1+z_1-z_2\}$
\end{tabular}

\item Branch $\mc S_{13}$: $3$ residues for finite $z_1$, $z_2$:

\begin{tabular}{ll}
XIII1\; $(0,0)\,,\; \{z_1,z_2\}$ & XIII2\; $(0,1)\,,\; \{z_1,z_1+z_2-1\}$
\\
XIII3\; $(1,0)\,,\; \{z_2,z_1+z_2-1\}$
\end{tabular}

\item Branch $\mc S_{14}$: $4$ residues for finite $z_1$, $z_2$:

\begin{tabular}{ll}
XIV1\; $(0,0)\,,\; \{z_1,z_2\}$ & XIV2\; $(0,1)\,,\; \{z_1,z_2-1\}$
\\
XIV3\; $(1,0)\,,\; \{z_1-1,z_2\}$ & XIV4\;$(1,1)\,,\; \{z_1-1,z_2-1\}$
\end{tabular}
\end{itemize}

Apparently, there are $64$ residues at finite loci. However, many of
them are redundant:

First, for branch $\mc S_1$, there are three polynomials $1+z_1$,
$z_2$ and $1+z_1-\chi z_2$ vanishing at
the residue $(-1,0)$. So it is a degenerate residue and there are three
different values of it, depending on the combination of the three
polynomials. Namely, we have three residues I2, I3 and I4 at
$(-1,0)$. However, by (\ref{triple_lemma}) in Lemma 1, the sum of the
three residues is zero. So we can eliminate the residue I4. Similarly, we
remove VIII4.

Second, the $14$ branches intersect with each other and many residues are
corresponding to the same value of the loop momenta. For example, it is
clear that, at the
residues I1 and II5, the ISP values are the same,
\begin{equation}
  \label{eq:8}
  \tilde\ell_1 \cdot \omega= \frac{t}{2}\;, \quad 
  \tilde\ell_2 \cdot \omega=0\;, \quad 
  \tilde\ell_3 \cdot \omega=0\;, \quad 
  \tilde\ell_1 \cdot k_4=0\;, \quad 
  \tilde\ell_3 \cdot k_4=0\;, \quad
  \tilde\ell_1 \cdot k_1=\frac{t}{2}\;, \quad 
  \tilde\ell_3 \cdot k_1=\frac{t}{2}\;.
\end{equation}
So they correspond to the same point in loop momenta space. Then we
can remove residue II5. Sometimes, the parametrization apparently blows up at one
residue. Then we can consider the projective space of loop momenta and
study the coincident residues. For example, residues I5 and IV3 
correspond to the same point in the projective loop momenta space. Exhausting
all these obvious coincident residues, we further remove $35$ residues. 

Third, there are more subtle coincident residues. For example, on the
branch $\mc S_1$, the
limit $z_1\to-1, z_2\to 0$ does not correspond to one point in the
projective loop momenta space, since the limit depends on the path
approaching $(-1,0)$. Explicitly, we can check that for the small
contour defined in (\ref{local_residue}) with the denominators choice
of I2, the corresponding loop
momenta approach that of residue V1. Similarly, (I3, V3), (VIII2,
II2), (VIII3, XII3) are also coincident residues in this pattern. So we can further
remove $4$ residues.

Finally, we have $23$ independent residues and choose them to be,
\begin{align}
  \mc R = {} &
  \{\text{I1, I2, I3, I5, I6, I7, I8, II1, II3, II4, II6, III1, III3, V2,}
  \nonumber \\ & \;\;\;
  \text{VIII1, VIII2, VIII3, VIII5, VIII6, IX3, IX4, IX6, XI2}\}\;.
 \label{residue_list}
\end{align}

Then we calculate the values of all terms in the numerator on these
residues. The non-degenerate residues can be calculated directly by
the Jacobian matrix, while the degenerate residues are calculated by
the transformation law (\ref{transformation_law}). We use a program
powered by the algebraic geometry software Macaulay2 \cite{M2}. Actually we
calculated the values on all $64$ residues, and explicitly verified the
relation (\ref{triple_lemma}) and that for each coincident residue pair,
the values are the same.  

Now that we have a minimal basis of triple box residues we can transform
\eqref{AUGDECA} into an algeberaic equation for the master integral
coefficients. We simply expand the master integrals, namely 
$\IP[1]$, $\IP[(\tilde\ell_1+k_4)^2]$ and $\IP[(\tilde\ell_3-k_4)^2]$, onto
the leading singularity cycles and calculate the 23 residues for each of them,
\begin{gather}
\Res_{\{\mc R\}}\,\IP[1] = \frac{1}{\chi^3 s_{12}^{10}}
\{-1,1,-1,0,0,0,-1,-1,0,0,-1,1,0,1,1,-1,1,0,0,0,0,-1,1\}\;, \\
\Res_{\{\mc R\}}\,\IP[(\tilde\ell_1+k_4)^2] = 
\frac{1}{\chi^2 s_{12}^9} 
\{0,1,0,0,0,0,-1,-1,0,0,0,1,0,1,0,-1,0,0,0,0,0,0,1\}\;, \\
\Res_{\{\mc R\}}\,\IP[(\tilde\ell_3-k_4)^2] = 
\frac{1}{\chi^2 s_{12}^9} 
\{0,1,0,0,1,-1,0,0,0,0,0,0,-1,0,0,-1,0,0,-1,0,0,0,0\}\;,
\end{gather}
where the values are listed in the same order as (\ref{residue_list}).

\subsection{Integral Reduction Identities}
Consistency of the unitarity procedure of course requires that vanishing
Feynman integrals must also have vanishing deca-cuts. However, promotion 
of real slice integrals to arbitrary multidimensional contour integrals 
implies that certain integral relations eventually fail to hold. Therefore it 
is necessary to constrain the integration contours and thus demand that any 
valid integral identity is preserved under replacement of contours, which is
to say,
\begin{align}
\mc I_1 = \mc I_2 \;\LRa\; \cut(\mc I_1) = \cut(\mc I_2)\;.
\end{align}
The origin of the integral reduction identities used for projection onto
master integrals is simple. In order to understand this better let us briefly 
consider the general integrand numerator polynomial which can be parametrized 
in terms of seven irreducible scalar products and naturally splits into 
spurious and nonspurious parts,
\begin{align}
N = {} & 
\sum_{\{\alpha_1,\dots,\alpha_7\}}c_{\alpha_1\cdots\,\alpha_7}
(\tilde\ell_1\cdot k_4)^{\alpha_1}
(\tilde\ell_2\cdot k_1)^{\alpha_2}
(\tilde\ell_3\cdot k_4)^{\alpha_3}
(\tilde\ell_3\cdot k_1)^{\alpha_4} \nn \\[-1mm]
& \qquad\qquad\times
(\tilde\ell_1\cdot\omega)^{\alpha_5}
(\tilde\ell_2\cdot\omega)^{\alpha_6}
(\tilde\ell_3\cdot\omega)^{\alpha_7} \nn \\[2mm]
= {} &
\sum_{\{\alpha_1,\dots,\alpha_4\}}
(\tilde\ell_1\cdot k_4)^{\alpha_1}
(\tilde\ell_2\cdot k_1)^{\alpha_2}
(\tilde\ell_3\cdot k_4)^{\alpha_3}
(\tilde\ell_3\cdot k_1)^{\alpha_4} \nn \\
& \qquad\qquad
\times\Big\{
c_{\alpha_1\cdots\,\alpha_40}^{\text{NS}}+
c_{\alpha_1\cdots\,\alpha_41}^{\text{NS}}
(\tilde\ell_1\cdot\omega)(\tilde\ell_2\cdot\omega) \nn \\[1mm]
& \qquad\qquad\qquad\quad
+c_{\alpha_1\cdots\,\alpha_40}^{\text{S}}(\tilde\ell_1\cdot\omega)+
c_{\alpha_1\cdots\,\alpha_41}^{\text{S}}(\tilde\ell_2\cdot\omega)+
c_{\alpha_1\cdots\,\alpha_42}^{\text{S}}(\tilde\ell_3\cdot\omega)\Big\}
\end{align}
where the additional orthogonal direction is represented by
\begin{align}
\omega\equiv\frac{1}{2s_{12}}\left(
\aAvg{2}{3}{1}\aAvg{1}{\gamma^\mu}{2}-
\aAvg{1}{3}{2}\aAvg{2}{\gamma^\mu}{1}
\right)\;.
\end{align}
The maximum powers of the coefficients are restricted by renormalizability
requirements and completion of the integrand reduction is achieved by 
multivariate polynomial division with respect to a Gr{\"o}bner basis 
constructed form the ten inverse propagators using the Mathematica package 
{\tt BasisDet} \cite{Zhang:2012ce} written by one of the present authors. The 
final integrand then consists of 199 spurious and 199 nonspurious terms.

Although spurious terms are of course not present in integrated expressions,
they play a vital role at the level of the integrand and for the winding
numbers. We require that all parity odd vanishing integral identities continue
to hold after pushing integration contours into $(\C^4)^{\otimes 3}$. This
leads to nontrivial constraints on the winding numbers. Even powers of
spurious ISPs are reducible and can be expressed in terms of other scalar
products by means of Gram determinant identities for four-dimensional momenta.
Another equivalent strategy is to identify the full variety of Levi-Civita
insertions that appear in integral reduction, after using momentum
conservation. 

The fact that total derivatives vanish upon integration allows us to add such 
terms to the integrand and thereby produce a vast set of relations among 
integrals known as integration-by-parts identities of the form
\begin{align}
\IP\bigg[\pd{v^\mu}{\ell_i^\mu}\bigg] = 
\int\frac{d^D\ell_1}{(2\pi)^D}
\int\frac{d^D\ell_2}{(2\pi)^D}
\int\frac{d^D\ell_3}{(2\pi)^D}
\pd{}{\ell_i^\mu}\frac{v^\mu}{\prod_{k=1}^{10}
D_k^{a_k}(\ell_1,\ell_2,\ell_3)} = 0\;.
\end{align}
In order to ensure validity of the unitarity method each nontrivial identity
gives rise to a constraint requiring that the procedure yields a vanishing 
coefficient for the additional term as this is not true automatically contour 
by contour. In the four-point massless planar triple box this amounts 199
integral identities corresponding to the nonspurious part of the integrand. We
include some of those identities for illustration purposes in
Appendix~\ref{IBPSAPP}. For most families of Feynman integrals the relations
are often quite involved to obtain. We refer to the public computer codes 
{\tt FIRE} \cite{Smirnov:2013dia} and {\tt Reduze}
\cite{vonManteuffel:2012np}, for more information on how to generate the
identities in practice.

\subsection{Unique Master Integral Projectors} 
Based on the discussion of the previous sections, we compute the residues from
the deca-cut to resolve all constraints from integration-by-parts identities
and vanishing of all spurious terms upon integration, and organize them as a
homogeneous system of equations with a matrix $\tilde{M}$ of size 
$398\times 23$. It turns out that $\operatorname{rank}\tilde{M} = 20$. In
other words, the integration contours are subject to only 20 constraints in
order to yield a valid unitarity procedure. 

Overall we are therefore left with three independent contour weights that are
not fixed by integral reduction consistency requirements, exactly matching the
number of planar triple box master integrals. This pleasant freedom ensures
that we can derive projectors or master contours that normalize one
master integral to unity and set the remaining two to zero. We therefore
extend the matrix $\tilde{M}$ to a $401\times 23$ matrix $M$ with the residues
of the master integrals and write down equations for the projectors, 
\begin{align}
M\Omega_1 = (0,\dots,0,1,0,0)^T\;, \quad 
M\Omega_2 = (0,\dots,0,0,1,0)^T\;, \quad
M\Omega_3 = (0,\dots,0,0,0,1)^T\;.
\end{align} 
Because $\operatorname{rank} M = 23$ and $\Omega_i$ is a 23-dimensional
vector, the solutions for the three projectors are unique. The result is
obtained by standard linear algebera. The operation of replacing the original
integration contour by a master contour in the augmented deca-cut thus
extracts the coefficient of the corresponding master integral in the basis
decomposition of the three-loop amplitude.

The three solutions for the contour weights are denoted $\Omega_1$, $\Omega_2$
and $\Omega_3$ and correspond to extracting master integral $\IP[1]$,
$\IP[(\tilde\ell_1+k_4)^2]$ and $\IP[(\tilde\ell_3-k_4)^2]$ respectively.
Being integer numbers of up to a common constant which corresponds to an
irrelevant overall normalization, the weights may be interpreted geometrically
as winding numbers of the global poles,
\begin{align}
\Omega_1 = {} & \frac{1}{8}\chi^3 s_{12}^{10}
\{-1,0,-2,0,1,1,0,0,1,-1,-1,0,1,0,1,0,2,0,-1,1,-1,-1,0\} \;, \\[1mm] 
\Omega_2 = {} & \frac{1}{4}\chi^2 s_{12}^9
\{0,1,2,-1,-2,-1,0,-1,-1,1,0,0,-1,1,0,-1,-2,1,2,-1,1,0,0\} \;, \\[1mm]
\Omega_3 = {} & \frac{1}{4}\chi^2 s_{12}^9
\{1,-1,-2,3,3,0,-2,1,0,0,1,2,0,-1,-1,1,2,-3,-3,0,0,1,0\} \;.
\end{align}
The associated master contours $\mc M_i$ can be constructed explicitly as
linear combinations of weighted infinitesimal toroidal surfaces encircling the
23 global poles. According to \eqref{residue_list}, each master contour only receives
contributions from a small subset $\Lambda$ of the 14 on-shell branches after
removal of redundant residues. Therefore we can apply the decomposition 
$\mc M_i = \sum_{k\in\Lambda} \mc M_{i;k}$ for 
$\Lambda = \{\text{I,II,III,V,VIII,IX,XI}\}$.
In this notation, weights are kept implicit in the contours. The master
integral coefficients can then be written schematically in the very compact
form 
\begin{align}
\label{MASTERFORMULAE}
C_i = \frac{1}{(2\pi i)^2}\sum_{k\in\Lambda}\oint_{\mc M_{i;k}}\!
\frac{dz_1\wedge dz_2}{J_k(z_1,z_2)}
\sum_{\substack{\text{helicities}\\\text{particles}}}
\prod_{j=1}^8 A_{(j)}^\tree(z_1,z_2)\big|_{\mc S_k}\;.
\end{align}
This formula completes our derivation of the triple box master integral
coefficients.

\section{Examples}
In this section we apply the master integral coefficient formulae to the
triple box contribution to three-loop four-point gluon amplitudes with
specific helicity configurations and massless kinematics. We only consider
deca-cuts in the $s$-channel as contributions in the $t$-channel can be
obtained in a completely similar manner by accounting for cyclic permutation
of external legs. 

The starting point for the computation is the intermediate state sum over the
product of eight tree-level amplitudes that arise by cutting the numerator
function $N$ on-shell on each branch,
\begin{align}
N|_{\mc S_i} = \sum_{\substack{\text{helicities}\\\text{particles}}}
\prod_{j=1}^8 A^\tree_{(j)}(z_1,z_2)\big|_{\mc S_i}\;.
\end{align}
The expression is summed over all possible internal helicity states and
distinct configurations of gluons, fermions and scalars inside the loops. It
may be obtained quite easily in $\mc N = 0,1,2,4$ supersymmetric Yang-Mills
theories by superspace techniques \cite{Bern:2009xq,Sogaard:2011pr} or simply
by direct computation in a generic theory from the distributions of
holomorphic and antiholomorphic vertices (see Appendix \ref{TREESMP}) and all
possible flavour configurations. For now we exploit that the triple box
tree-level data has already been calculated for $n_f$ fermion and $n_s$
complex scalar flavours in the adjoint representation and refer to 
\cite{Badger:2012dv} for more details.

We examine the alternating helicity configuration $-+-+$. In this case, the
23 independent residues computed from the Laurent expansions of the products
of tree-level amplitudes explicitly read 
\begin{align}
\Res_{\{\mc R\}}
\bigg(
\frac{dz_1\!\wedge\!dz_2}{J(z_1,z_2)}\!
\sum_{\substack{\text{helicities}\\\text{particles}}}
\prod_{j=1}^8 A^{-+-+}_{(j)}\bigg) = {} &
\frac{1}{\chi^3 s_{12}^{10}}
\{1,-1,1,-r_1,0,r_1,1-r_1,1-r_1,r_1,-r_1,1,-1, \nn \\[-2mm] & \;
0,-1,-1+r_1,1-r_1,-1+r_1,0,r_2,0,0,1-r_1,-1\}\;,
\label{RESIDUES-+-+}
\end{align}
where the integrand is understood to be evaluated on the branch on which the
residue $\mc R_i$ actually resides. The new variables $r_1$ and $r_2$ introduced
in \eqref{RESIDUES-+-+} are defined in the following way,
\begin{align}
r_1\equiv {} & 
\frac{\chi(1+\chi^2)(4-n_f)+2\chi^2(3-n_s)}{(1+\chi)^4}\;, \\[2mm]
r_2\equiv {} & 
\frac{\chi}{(1+\chi)^4}\Big(
8(1-2\chi)(3-n_s)(4-n_f)-(13-(24-\chi)\chi)(4-n_f)
-2\chi(3-n_s)\\
& \qquad\qquad\quad 
+(1-(4-\chi)\chi)(n_f(3-n_s)^2-2(4-n_f)^2)-2(1-2\chi)(4-n_f)^2\Big)
\;. \nn
\end{align}
In a supersymmetric theory we may rewrite $n_f = \mc N$ and $n_s = \mc N-1$
where $\mc N$ counts the number of supercharges. With this in mind, the
residues were resolved into a form with coefficients that vanish for $\mc N =
4$ and also $\mc N = 2$. Summation of all residues with weights according to
the three master contours now yields the normalized master integral
coefficients,
\begin{align}
\hat C_1^{-+-+} = {} &
-1+(4-n_f)\frac{\chi}{(1+\chi)^2}-2(1+n_s-n_f)\frac{\chi^2}{(1+\chi)^4} \nn \\ &
-(2(1-2n_s)+n_f)(4-n_f)\frac{(1-2\chi)\chi}{4(1+\chi)^4} \nn \\[1mm] &
-(n_f(3-n_s)^2-2(4-n_f)^2)\frac{(1-(4-\chi)\chi)\chi}{8(1+\chi)^4}\;, \\[1mm]
\hat C_2^{-+-+} = {} &
-(4-n_f)\frac{1}{s_{12}(1+\chi)^2}+2(1+n_s-n_f)\frac{\chi}{s_{12}(1+\chi)^4} \nn \\ &
+(2(1-2n_s)+n_f)(4-n_f)\frac{1-2\chi}{s_{12}(1+\chi)^4} \nn \\[1mm] &
+(n_f(3-n_s)^2-2(4-n_f)^2)\frac{1-(4-\chi)\chi}{2s_{12}(1+\chi)^4}\;, \\[1mm]
\hat C_3^{-+-+} = {} &
-(2(1-2n_s)+n_f)(4-n_f)\frac{3(1-2\chi)}{2s_{12}(1+\chi)^4} \nn \\ &
-(n_f(3-n_s)^2-2(4-n_f)^2)\frac{3(1-(4-\chi)\chi)}{4s_{12}(1+\chi)^4}\;.
\end{align}
The case of helicities $--++$ simply gives
\begin{align}
\hat C_1^{--++} = -1\;, \quad 
\hat C_2^{--++} = 0\;, \quad
\hat C_3^{--++} = 0\;.
\end{align}
The coefficients displayed here therefore clearly reproduce the known result
in $\mc N = 4$ super Yang-Mills theory \cite{Bern:2005iz} and also agree with
the general expressions for any number of adjoint fermions and scalars
\cite{Badger:2012dv}. We finally verified the last independent
helicity configurations, namely $-++-$. 

\section{Conclusion}
In this paper we have extended four-dimensional univariate maximal unitarity
\cite{Kosower:2011ty} to cuts that define multidimensional algebraic
varieties in any gauge theory. As an example, we studied the application to
three-loop amplitudes. In maximal unitarity one cuts the maximum number of
propagators on-shell by replacement of real slice integrations by
multidimensional contours encircling the global poles of the loop integrand.
Each residue comes with a weight or winding number in order for the procedure
to conform with integral reduction identities from integration-by-parts
relations and parity-odd terms that vanish upon integration. 

The technique was demonstrated explicitly for the planar triple box with four
external massless legs. We obtained unique and very simple projectors for all
master integrals from just 23 finite multivariate residues of the maximally
cut loop integrand. We proved by the Global Residue Theorem that all residues
at infinity are linear combinations of finite residues and therefore not
needed in the computation. Also, we worked out master integral coefficients
for all independent helicity configurations in any gauge theory and found
exact agreement with recently obtained results \cite{Badger:2012dv}.

Our method is completely general and we expect it to apply to multivariate
residues in more than two complex variables, for instance in the $L$-loop
ladder topology and similar planar and nonplanar diagrams. It also establishes
an initial foundation for extracting master integral coefficients below the
leading singularity at the two-loop level and therefore provides a step
towards full automation. A complete calculation also requires terms of 
$\mc O(\epsilon)$ beyond four-dimensional cuts.

We also give directions for research in the future. It would be very
beneficial to better understand the simplicity of the projectors at two and
more loops. Guided by recent progress in integrand-level reduction
\cite{Badger:2013gxa,Mastrolia:2013kca}, we find it very interesting to
generalize the maximal unitarity method to $D = 4-2\epsilon$ dimensions where
amplitudes are cut into products of six-dimensional trees. In that case, the
polynomial ideal generated by a set of propagators is prime and hence there is
only one branch of the unitarity cut. This will maybe make it easier to
resolve the redundancy among residues. Scattering processes with more than
four external particles are also important to examine. Higher-point integrals
suffer from a more complicated set of reduction identities due to dependence
on a broader variety of kinematical invariants. Planar and nonplanar double
box integrals with five massless external legs probably hint the simplest
extension of maximal unitarity in this direction because their hepta-cut
equations are similar to those of the four-point case. As a general
mathematical tool we believe that multivariate residue calculations will also
prove useful broadly in the the study of supersymmetric Yang-Mills amplitudes
and evaluation of master integrals. We also expect that computational
algebraic geometry can offer important insight in the topological information
of algebraic varieties from unitarity cuts such as degeneracies under specific
kinematics. Hopefully some of these questions will be addressed soon.

\acknowledgments
We are grateful to Emil Bjerrum-Bohr, Poul Henrik Damgaard and Rijun Huang for
many useful discussions. It a pleasure to thank Simon Badger and Hjalte
Frellesvig for reading the paper in draft stage and for sharing material on
integration-by-parts identities. We also thank Simon Caron-Huot for comments
on the manuscript. MS acknowledges the theoretical elementary particle
physics group at UCLA and in particular Zvi Bern for hospitality during the
initial stages of this work. YZ expresses gratitude to Mingmin Shen for the
help on concepts of algebraic geometry and comments on the mathematical
section of this paper. The work of YZ is supported by Danish Council for
Independent Research (FNU) grant 11-107241.

\appendix
\clearpage
\section{Residues at Infinity}
\label{poles_at_infinity}

In this Appendix, we use the global residue theorem (GRT) (\ref{GRT}) to prove
that for the planar triple box, all residues at infinity are linearly related
to residues at finite loci. 

Since the branches $\mc S_8,\dots,\mc S_{14}$ are complex conjugate
of $\mc S_1,\dots,\mc S_7$, we only present the residues at infinity
for $\mc S_1,\dots,\mc S_7$. We start from the simplest cases $\mc
S_5$, $\mc S_6$ and $\mc S_7$, then $\mc S_1$ and $\mc S_2$, and finally the
complicated cases $\mc S_3$ and $\mc S_4$.

We extend the space $\mathbb C^2: (z_1, z_2)$ to $\mathbb {CP}^2:
[w_0, w_1, w_2]$ as $z_1=w_1/w_0$ and $z_2=w_2/w_0$. The points in
$\mathbb {CP}^2$ with $w_0=0$ are called points at infinity.
\bigskip
\begin{figure}[h]
\centering
\includegraphics[width=1\textwidth]{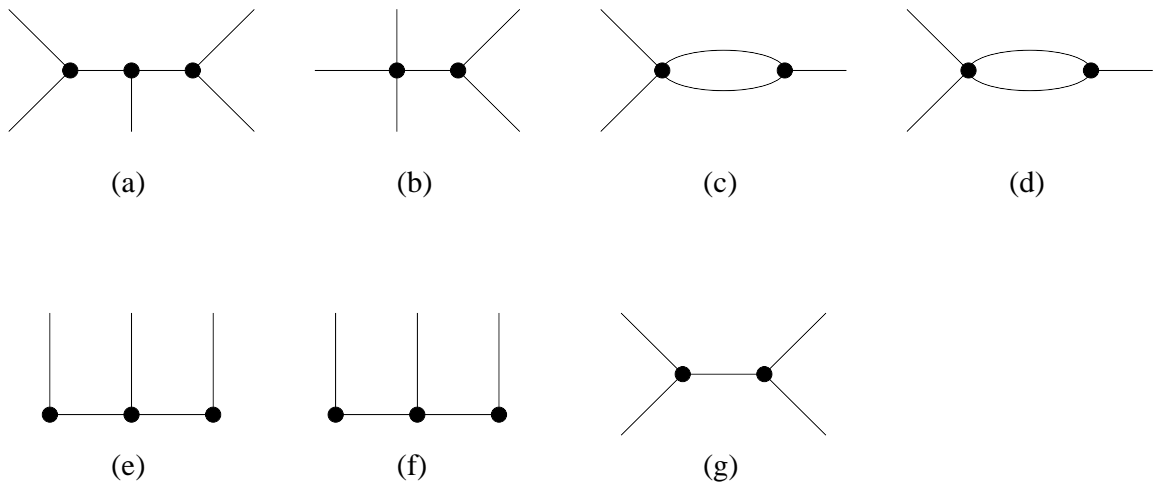}
\caption{{\it Infinity diagrams}: Inside one diagram, each line (curve)
  represents a vanishing polynomials in the denominators and each
  black dot represents a residue at infinity.}
\label{Infinite_diagrams}
\end{figure}

\begin{itemize}
\item Branch $\mc S_5$. The denominators in the Jacobian are $z_1,
  1+z_1-z_2, z_2$. Furthermore, the differential form may blow up at
  $w=0$. So we list the vanishing polynomials in projective variables,
  \begin{equation}
    \label{eq:16}
    w_1\,, \quad w_0+w_1-w_2\,, \quad w_2\,, \quad w_0\,.
  \end{equation}
It is clear that $w_1$ and $w_0$ both vanish at $[0,0,1]$,
$w_0+w_1-w_2$ and $w_2$ vanish at $[0,1,1]$, $w_2$ and $w_0$ 
vanish $[0,1,0]$. So there are
$3$ residues at infinity. Each of these residues are linear
combinations of the residues at finite loci, by GRT. For example, for
the residue with the denominators,
\begin{equation}
  \label{eq:17}
  [0,0,1]\,, \quad \{w_1, w_0\}\,,
\end{equation}
we can choose two divisors on $\mathbb{CP}^2$, $D_1=\{w_1=0\}$,
$D_2=\{w_0(w_0+w_1-w_2)w_2=0\}$. Then $D_1\cap D_2=\{[0,0,1],
[1,0,1], [1,0,0]\}$. Only the first one is a residue at infinity,
while the rest two are located at the residues labeled as V3 and
V2. Furthermore, near $[1,0,1]$, locally $D_1$ and $D_2$ reduce to
the hypersurfaces $\{z_1=0\}$ and $\{1+z_1-z_2=0\}$, which are just
the denominators defined in V3. Similar,  $D_1$ and $D_2$ reduced to
the denominators defined in V2 near $[1,0,0]$. Hence, by GRT, the
value of the residue at $[0,0,1]$ is the opposite to the sum of
values from V2 and V3. So we do not need to consider this
residue. Similarly, the other residues at infinity for this branch can
also be ignored.

We find it is helpful to illustrate the residues at infinity by a
sketch diagram in fig. \ref{Infinite_diagrams}. (e). We use lines (curves) to
represent vanishing polynomials and
intersection points (black dots) for residues at infinity. Residues at
finite loci are not shown. We call such a diagram, {\it infinity
diagram}. For branch $\mc S_5$, we use a horizontal
line for the polynomial $w_0$, and three vertical lines for $w_1$,
$w_0+w_1-w_2$ and $w_2$. There are three intersecting points, which
are three residues at infinity. To use GRT for $ [0,0,1]$,
we split the diagram to two components, one with $w_1$, the other with
$w_0$ $w_0+w_1-w_2$ and $w_2$. These two components correspond to two
divisors. Because the two components have no other intersecting point
at infinity, by GRT, the residue $ [0,0,1]$ should be a linear combination of
residues at finite loci. This diagram makes the proof clear.

\item Branch $\mc S_6$. The vanishing polynomials in projective
  variables are,
  \begin{equation}
    \label{eq:18}
    w_1\,, \quad w_2\,, \quad -w_0 + w_1 + w_2\,, \quad w_0\,,
  \end{equation}
and there are three residues at infinity. Its {\it infinity diagram}, fig.
\ref{Infinite_diagrams}. (f), has the same
structure as that of $\mc S_5$, and by the same proof, all residues
at infinity are linear combinations of residues at finite loci.

\item Branch $\mc S_7$.  The vanishing polynomials in projective
  variables are,
  \begin{equation}
    \label{eq:18}
    -w_0 + w_1\,, \quad w_1\,, \quad -w_0 + w_2\,, \quad w_2\,, \quad w_0\,.
  \end{equation}
The first two polynomials and $w_0$ vanish at $[0,0,1]$, while the
third and fourth and $w_0$ vanish at $[0,1,0]$. There are two residues
at infinity and both are degenerate, which can give four values
depending on the choice of denominators. The {\it infinity diagram} is
shown in fig.
\ref{Infinite_diagrams} (g). Again, the topology of this diagram, is a tree. 
So for any residue at infinity, no matter how to choose the denominators, we
can always split the diagram into two and the two denominators are in
different components. Then GRT shown that all residues at infinity can be
ignored.

\item Branch $\mc S_1$. The vanishing polynomials in projective
  variables are,
\begin{equation}
  \label{eq:19}
  w_1\,, \quad w_0 + w_1\,, \quad w_2\,, \quad w_0 + w_2\,, \quad 
  w_0 + w_1 - \chi w_2\,, \quad w_0\,.
\end{equation}
The intersecting structure is more complicated. However, the
{\it infinity diagram},  shown in fig.
\ref{Infinite_diagrams}. (a),
is still a tree, so by using GRT several times, we can shown that all residues
at infinity can be
ignored.

\item Branch $\mc S_2$. The vanishing polynomials in projective
  variables are,
\begin{equation}
  \label{eq:19}
  w_1\,, \quad w_0 + w_1\,, \quad w_0 + w_1 + \chi w_1\,, \quad w_2\,, \quad  
  w_0 +w_2\,, \quad w_0\,.
\end{equation}
Similarly, the {\it infinity diagram},  shown in fig.
\ref{Infinite_diagrams}. (b),
is a tree, so all residues at infinity can be
ignored.

\item Branch $\mc S_3$. The vanishing polynomials in projective
  variables are,
\begin{equation}
w_1\,, \quad w_0 + w_1\,, \quad w_2\,, \quad w_0^2 + w_0 w_1 + \chi w_0 w_1 -
\chi w_1w_2\,, \quad w_0\,.
\end{equation}
 The structure of residues at infinity of this diagram is
 complicated. Note that the fourth polynomial is quadratic, and it has
 two intersection points with the hypersurface $\{w_0=0\}$. The first,
 second, fourth and fifth polynomials vanish at $[0,0,1]$. On the
 other hand, the third,
 fourth and fifth polynomials vanish at $[0,1,0]$. So the
 corresponding {\it infinity diagram}, shown in fig.
\ref{Infinite_diagrams}. (c), has the topology of a
 loop. Hence the direct proof by GRT does not work for this branch,
 and residues at infinity are not linear combination of residues
 at finite loci
 on this branch.

However, we show that the residues at infinity are actually the sum of
residues on
$\mc S_3$ and also other branches. It is clearly that by GRT, if we know all
the values of the residue at $[0,1,0]$, we get all the values of
residues at $[0,0,1]$. At $[0,1,0]$, there are three 
residues, corresponding to, $\{w_0 w_2, w_0^2 + w_0 w_1 + \chi w_0 w_1 - \chi
w_1 w_2\}$
$\{ (w_0^2 + w_0 w_1 + \chi w_0 w_1 - \chi w_1 w_2) w_0,w_2\}$ and $\{w_2 (
w_0^2 + w_0 w_1 + \chi w_0 w_1 - \chi w_1 w_2),w_0\}$. By Lemma
1 (\ref{triple_lemma}) (generalized version), the three residues sum
to zero, so we only need the first two. Furthermore, for the second
one, the two polynomials for the loop in the {\it infinity diagram} are
combined together. So it is possible to split the diagram into two
components intersecting at only one point. Then this residue is a
combination of residues at finite loci. Finally we only need to
consider the first one. Explicitly, we find that this residue
and the residue II4 (on branch $\mc S_2$) correspond to the same point in the
projective
loop momenta space. The values of this residue for all integrand terms
match these of residue II4. So this residue at infinity is not
needed. Therefore, no residue at infinity is needed for this branch. 

\item Branch $\mc S_4$.  The vanishing polynomials in projective
  variables are,
  \begin{equation}
    \label{eq:21}
    w_1\,, \quad w_0 + w_1\,, \quad w_2\,, 
    \quad w_0^2 + w_0 w_1 + \chi w_0 w_1 + w_1 w_2 + \chi w_1 w_2\,, \quad
w_0\,.
  \end{equation}
Again, the {\it infinity diagram}, shown in fig.
\ref{Infinite_diagrams}. (d), is a loop and so the direct proof by GRT
does not work. However, by the same analysis as that for $\mc S_3$, the
residues at infinity are actually the sum of residues on
$\mc S_4$ and also other branches. No residue at infinity is needed for this
branch. 
\end{itemize}

\clearpage
\section{Kinematical Configurations of the Planar Triple Box}
\label{TREESMP}
We include here valid kinematical configurations of three-point trees in the 
maximally cut planar triple box such that no external legs are neither 
holomorphically nor antiholomorphically collinear for generic momenta. The
diagrams are in one-to-one correspondence with the on-shell solutions and 
follow their enumeration described in the main text. For brevity we only
depict contributions that are not related to each other by parity conjugation.
In our conventions, black and white blobs denote MHV and googly-MHV vertices
respectively.
\vspace*{.5cm}
\begin{figure}[!h]
\bc
Solution $\mc S_1$ \\
\includegraphics[scale=0.425]{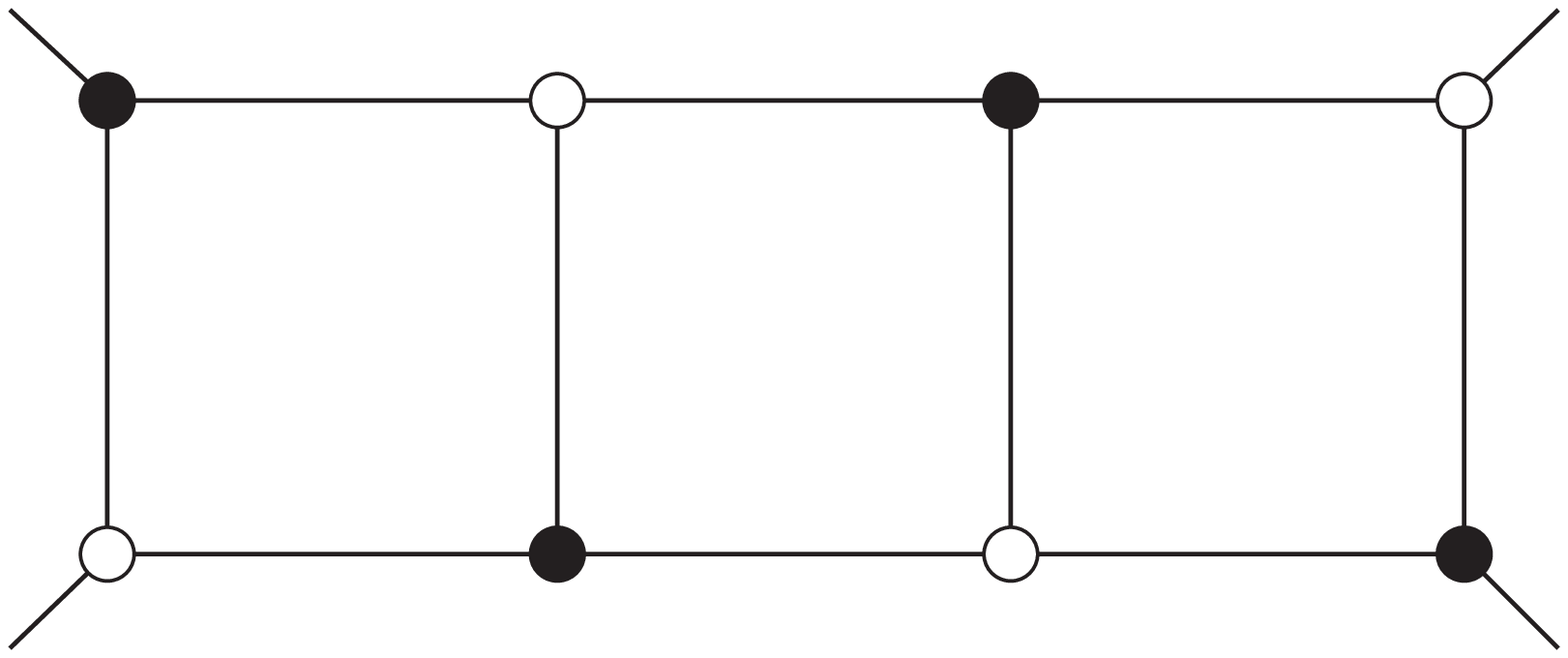}
\vspace*{.3cm}
\ec

\begin{minipage}{0.5\textwidth}
\bc
Solution $\mc S_2$ \\
\includegraphics[scale=0.425]{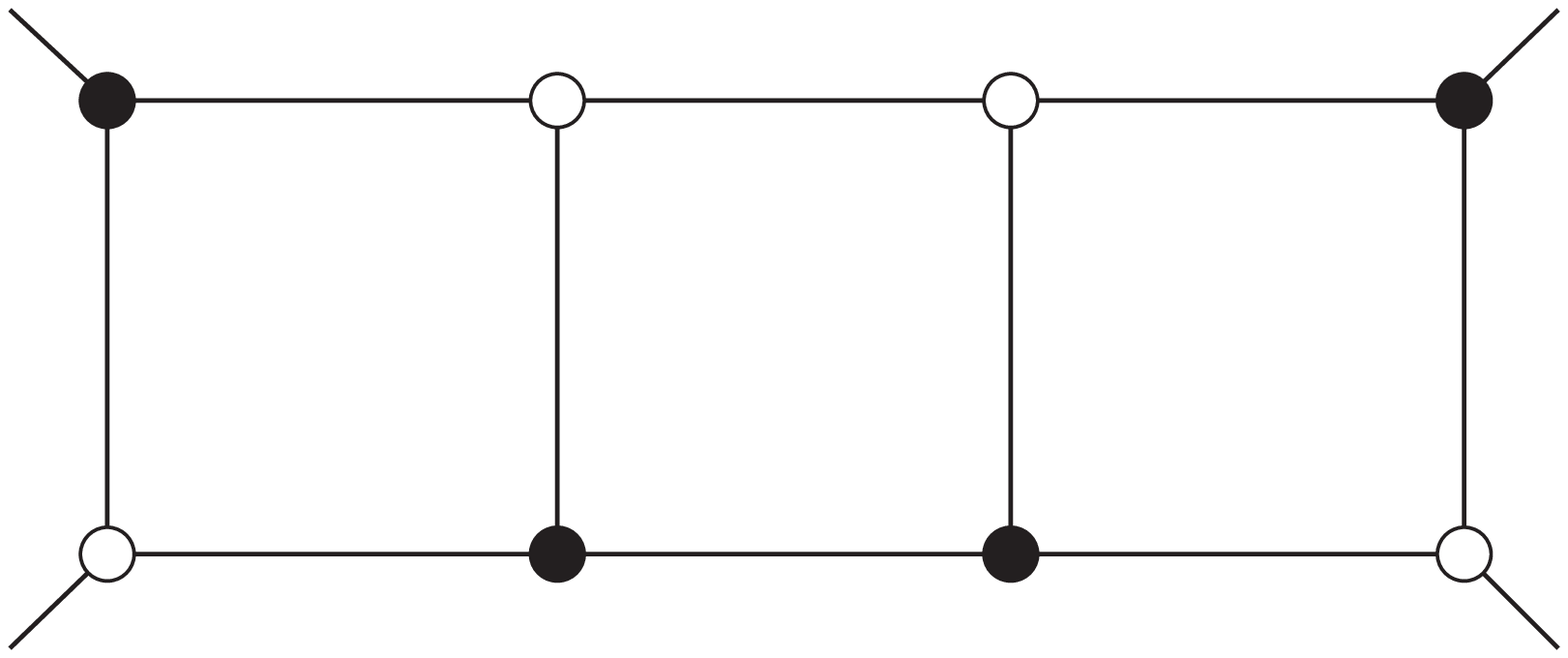}
\vspace*{.5cm}

Solution $\mc S_4$ \\
\includegraphics[scale=0.425]{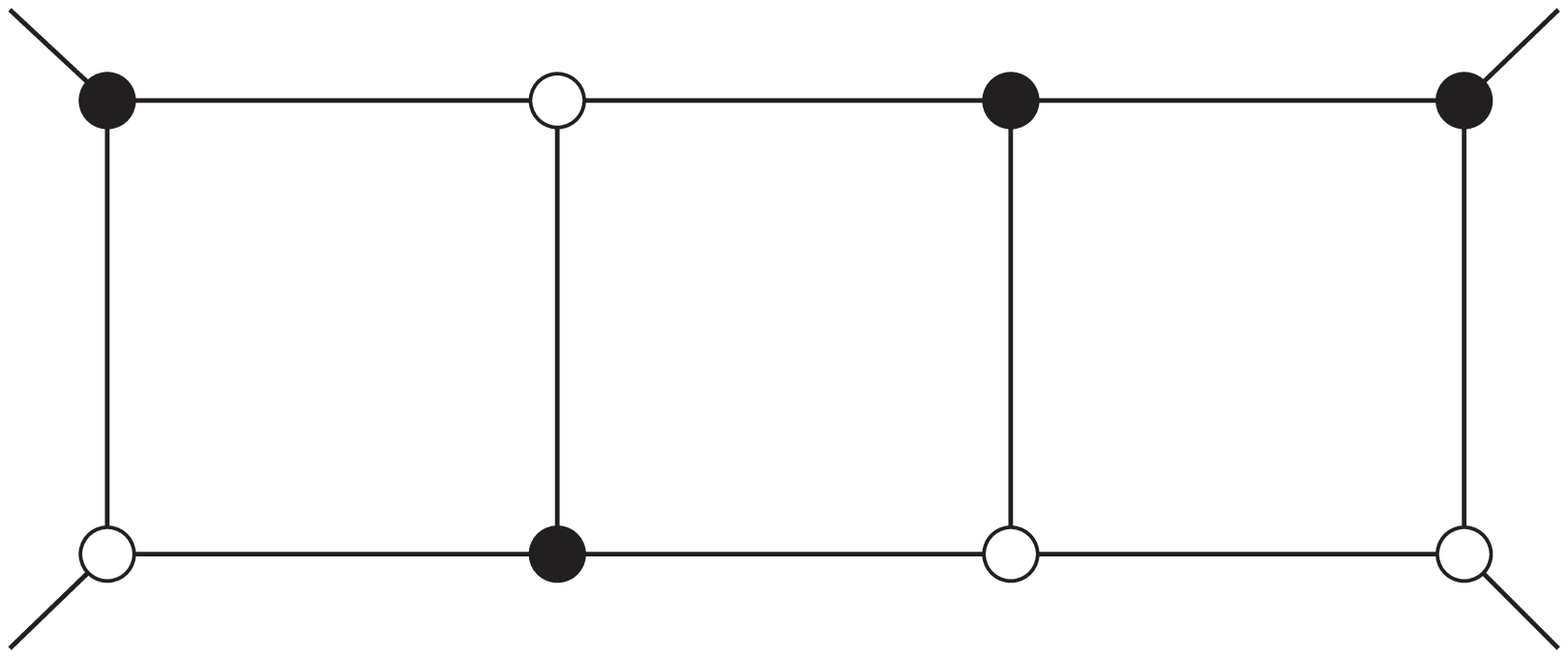}
\vspace*{.5cm}

Solution $\mc S_6$ \\
\includegraphics[scale=0.425]{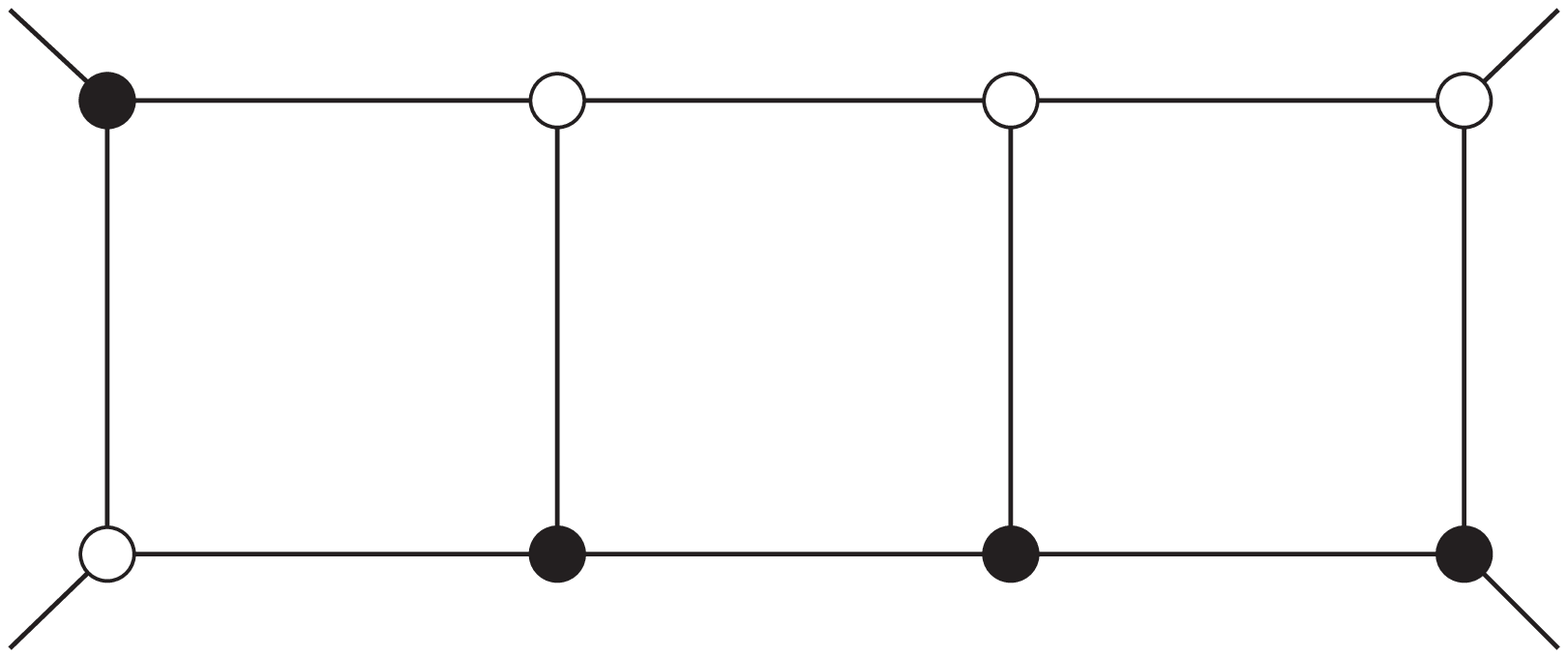}
\ec
\end{minipage}
\begin{minipage}{0.5\textwidth}
\bc
Solution $\mc S_3$ \\
\includegraphics[scale=0.425]{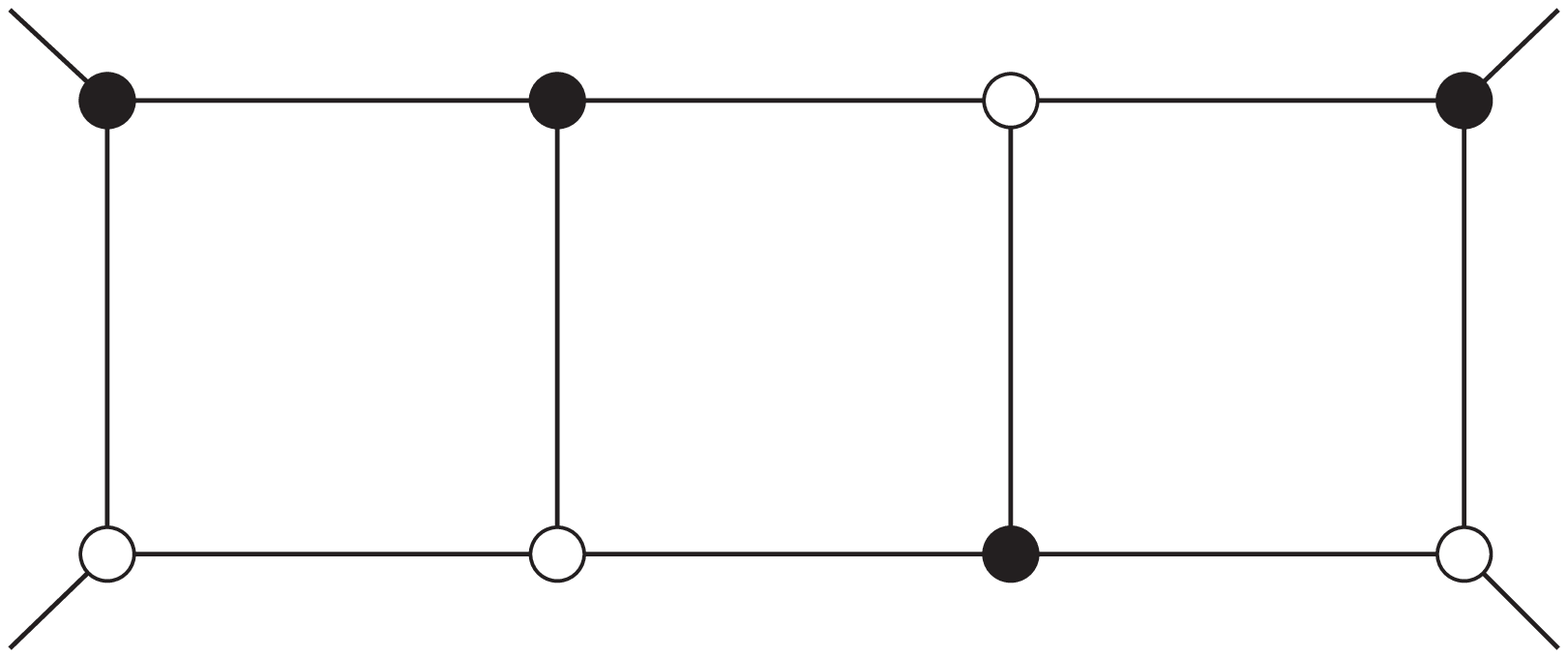}
\vspace*{.5cm}

Solution $\mc S_5$ \\
\includegraphics[scale=0.425]{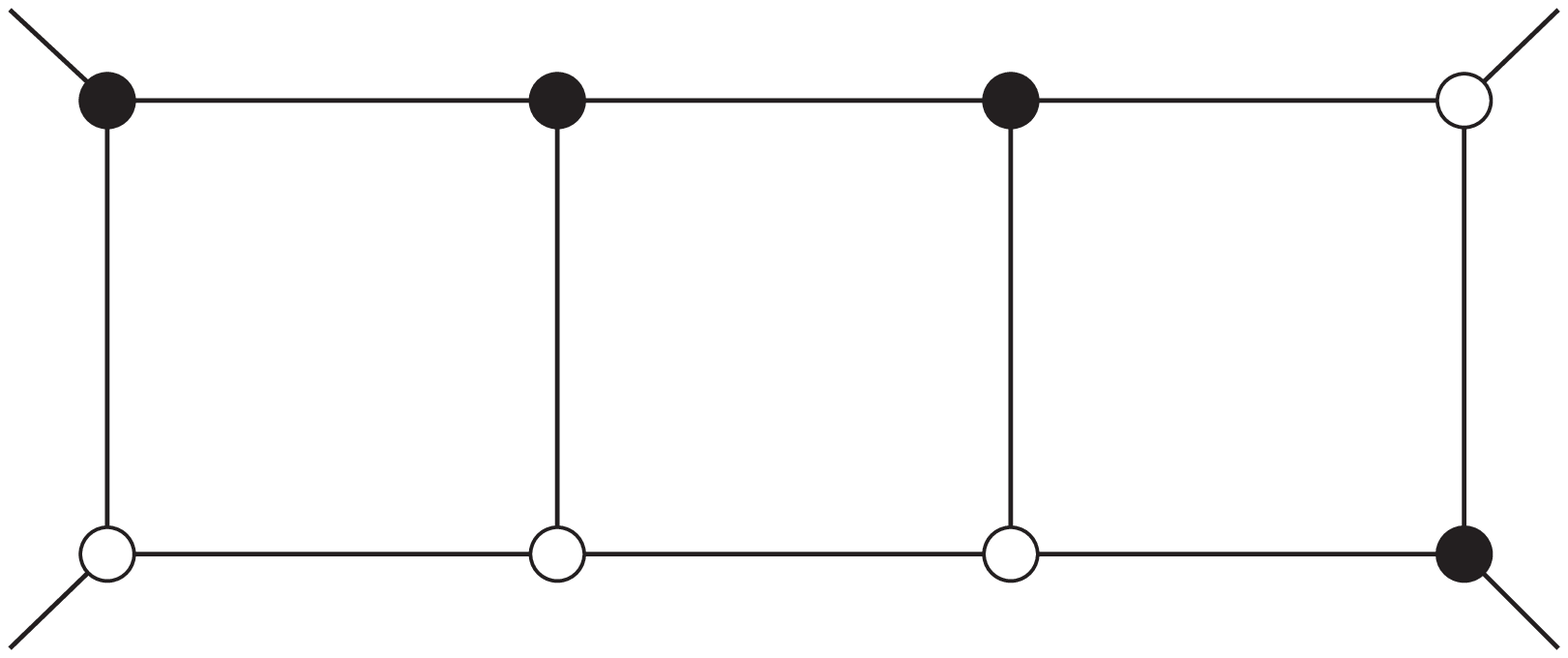}
\vspace*{.5cm}

Solution $\mc S_7$ \\
\includegraphics[scale=0.425]{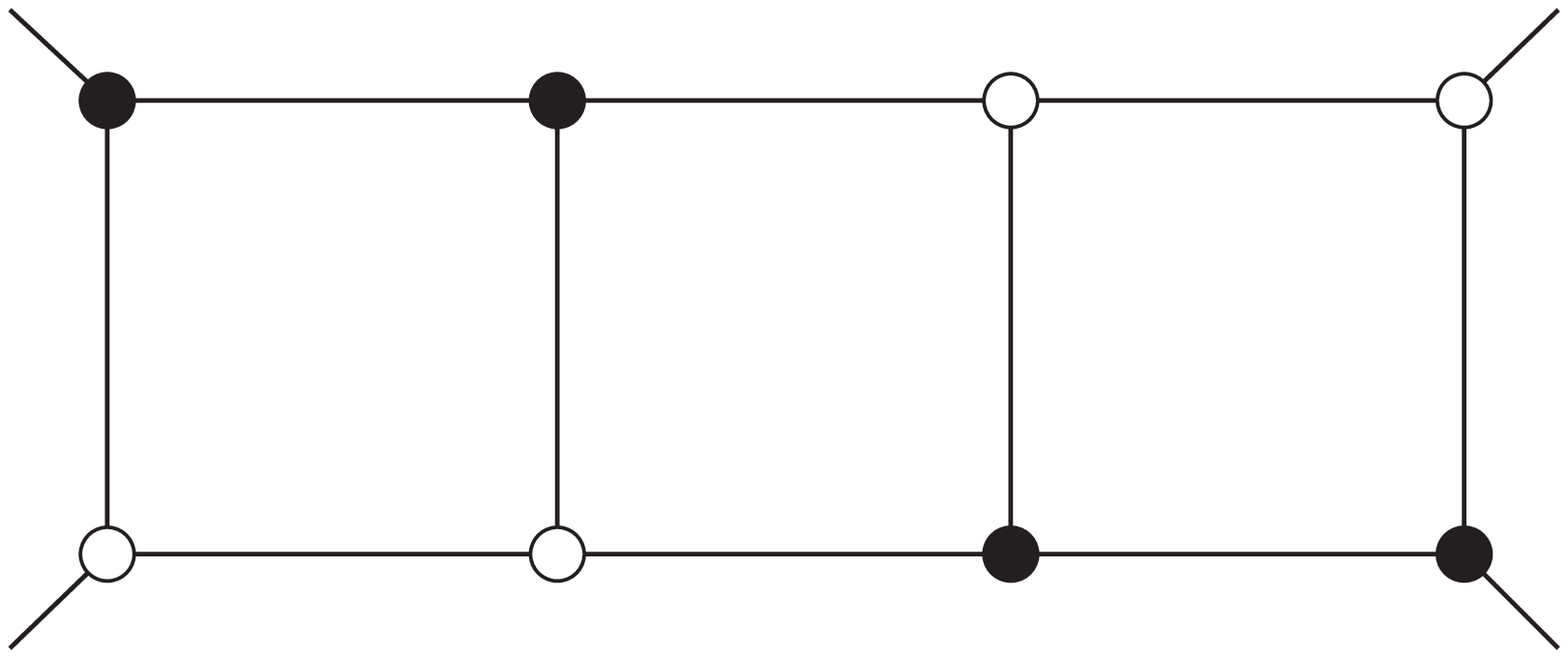}
\ec
\end{minipage}
\end{figure}

\clearpage
\section{Planar Triple Box Integration-By-Parts Identities}
\label{IBPSAPP}
We provide below a small subset of the four-dimensional integration-by-parts
identities used for the reduction onto master integrals of renormalizable
Feynman integrals with planar triple box topology. The relations are for
simplicity stated in terms of irreducible scalar products (ISPs) defined by 
$x_{ij}\equiv \tilde\ell_i\cdot v_j$ with momentum basis $v =
(k_1,k_2,k_4,\omega)$ so that the master integrals can be written
\begin{align} 
\IP[1]\;, \quad \IP[x_{13}]\equiv\frac{1}{2}\IP[(\tilde\ell_1+k_4)^2]+\cdots\;,
\quad \IP[x_{33}]\equiv-\frac{1}{2}\IP[(\tilde\ell_3-k_4)^2]+\cdots\;,
\end{align}
whereas the general integral is 
$\IP[x_{13}^{a_1}x_{21}^{a_2}x_{33}^{a_3}x_{31}^{a_4}x_{14}^{a_5}x_{24}^{a_6}]$.
Then we have for instance
\begin{align}
\IP[x_{31}] = {} & 
-\IP[x_{33}]+\cdots \\[2mm]
\IP[x_{13}^2] = {} & 
+\frac{1}{2}\chi s_{12}\IP[x_{13}]+\cdots \\[2mm]
\IP[x_{21}^4] = {} & +\frac{1}{8}\chi^3 s_{12}^3\IP[x_{13}]+\cdots \\[2mm]
\IP[x_{13}^2x_{21}x_{31}] = {} & 
+\frac{1}{32}\chi s_{12}^4\IP[1]
-\frac{1}{8}s_{12}^3\IP[x_{13}] \\[2mm]
\IP[x_{13}^2x_{21}] = {} & 
-\frac{1}{16}\chi s_{12}^3\IP[1]
+\frac{1}{4}s_{12}^2\IP[x_{13}]
-\frac{1}{4}\chi s_{12}^2\IP[x_{33}]+\cdots \\[2mm]
\IP[x_{13}x_{21}^4x_{33}] = {} & 
-\frac{1}{128}\chi(2+\chi)s_{12}^6\IP[1]
+\frac{1}{32}(2+\chi)s_{12}^5\IP[x_{13}]+\cdots\;, \\[2mm]
\IP[x_{13}^2x_{21}^3x_{33}] = {} & 
-\frac{1}{128}\chi(3+\chi(3+\chi))s_{12}^6\IP[1]
+\frac{1}{32}(3+\chi(3+\chi))s_{12}^5\IP[x_{13}]+\cdots 
\end{align}
Ellipses denote truncation at the maximal number of propagators. We also have
to consider integral reduction identites for terms in the integrand that
contain products of spurious scalar products $x_{14}$ and $x_{24}$. A few of
these relations read 
\begin{align}
\IP[x_{14}x_{24}] = {} &
+\frac{1}{8}\chi s_{12}^2\IP[1]
-\frac{1}{2}s_{12}\IP[x_{13}]
+\frac{1}{2}\chi s_{12}\IP[x_{33}]\cdots \\[2mm]
\IP[x_{21}x_{14}x_{24}] = {} &
-\frac{1}{16}\chi(1+\chi)s_{12}^3\IP[1]
+\frac{1}{4}(1+\chi)s_{12}^2\IP[x_{13}]+\cdots \\[2mm]
\IP[x_{33}x_{31}x_{14}x_{24}] = {} &
+\frac{1}{64}\chi^2s_{12}^4\IP[1]
-\frac{1}{8}\chi s_{12}^3\IP[x_{13}]
-\frac{3}{16}\chi s_{12}^3\IP[x_{33}]+\cdots \\[2mm]
\IP[x_{31}^2x_{14}x_{24}] = {} &
+\frac{1}{64}\chi s_{12}^4\IP[1]
-\frac{1}{16}(1-\chi)s_{12}^3\IP[x_{13}]
+\frac{3}{16}\chi s_{12}^3\IP[x_{33}]+\cdots
\end{align}
The remaining relations, which we do not include here, display the same
simplicity. Many high-rank integrals are actually reducible below ten
propagators, which means clearly means that they should vanish on the maximal
cut.


\clearpage
\begin{thebibliography}{99}
\bibitem{Britto:2004ap}
  R.~Britto, F.~Cachazo and B.~Feng,
  {\it New recursion relations for tree amplitudes of gluons,}
  Nucl.\ Phys.\ B {\bf 715}, 499 (2005)
  [hep-th/0412308].

\bibitem{Britto:2005fq}
  R.~Britto, F.~Cachazo, B.~Feng and E.~Witten,
  {\it Direct proof of tree-level recursion relation in Yang-Mills theory,}
  Phys.\ Rev.\ Lett.\  {\bf 94}, 181602 (2005)
  [hep-th/0501052].

\bibitem{Bern:2008qj}
  Z.~Bern, J.~J.~M.~Carrasco and H.~Johansson,
  {\it New Relations for Gauge-Theory Amplitudes,}
  Phys.\ Rev.\ D {\bf 78}, 085011 (2008)
  [arXiv:0805.3993 [hep-ph]].


\bibitem{Bern:1994cg}
  Z.~Bern, L.~J.~Dixon, D.~C.~Dunbar and D.~A.~Kosower,
  {\it Fusing gauge theory tree amplitudes into loop amplitudes,}
  Nucl.\ Phys.\ B {\bf 435}, 59 (1995)
  [hep-ph/9409265].

\bibitem{Bern:1994zx}
  Z.~Bern, L.~J.~Dixon, D.~C.~Dunbar and D.~A.~Kosower,
  {\it One loop n point gauge theory amplitudes, unitarity and collinear limits,}
  Nucl.\ Phys.\ B {\bf 425}, 217 (1994)
  [hep-ph/9403226].

\bibitem{Bern:1995db}
  Z.~Bern and A.~G.~Morgan,
  {\it Massive loop amplitudes from unitarity,}
  Nucl.\ Phys.\ B {\bf 467}, 479 (1996)
  [hep-ph/9511336].

\bibitem{Bern:1997sc}
  Z.~Bern, L.~J.~Dixon and D.~A.~Kosower,
  {\it One loop amplitudes for e+ e- to four partons,}
  Nucl.\ Phys.\ B {\bf 513}, 3 (1998)
  [hep-ph/9708239].

\bibitem{Britto:2004nc}
  R.~Britto, F.~Cachazo and B.~Feng,
  {\it Generalized unitarity and one-loop amplitudes in N=4 super-Yang-Mills,}
  Nucl.\ Phys.\ B {\bf 725}, 275 (2005)
  [hep-th/0412103].

\bibitem{Britto:2004nj}
  R.~Britto, F.~Cachazo and B.~Feng,
  {\it Computing one-loop amplitudes from the holomorphic anomaly of unitarity
  cuts,}
  Phys.\ Rev.\ D {\bf 71}, 025012 (2005)
  [hep-th/0410179].

\bibitem{Bern:2005hh}
  Z.~Bern, N.~E.~J.~Bjerrum-Bohr, D.~C.~Dunbar and H.~Ita,
  {\it Recursive calculation of one-loop QCD integral coefficients,}
  JHEP {\bf 0511}, 027 (2005)
  [hep-ph/0507019].

\bibitem{Bidder:2005ri}
  S.~J.~Bidder, N.~E.~J.~Bjerrum-Bohr, D.~C.~Dunbar and W.~B.~Perkins,
  {\it One-loop gluon scattering amplitudes in theories with N < 4
  supersymmetries,}
  Phys.\ Lett.\ B {\bf 612}, 75 (2005)
  [hep-th/0502028].

\bibitem{Britto:2005ha}
  R.~Britto, E.~Buchbinder, F.~Cachazo and B.~Feng,
  {\it One-loop amplitudes of gluons in SQCD,}
  Phys.\ Rev.\ D {\bf 72}, 065012 (2005)
  [hep-ph/0503132].

\bibitem{Britto:2006sj}
  R.~Britto, B.~Feng and P.~Mastrolia,
  {\it The Cut-constructible part of QCD amplitudes,}
  Phys.\ Rev.\ D {\bf 73}, 105004 (2006)
  [hep-ph/0602178].

\bibitem{Mastrolia:2006ki}
  P.~Mastrolia,
  {\it On Triple-cut of scattering amplitudes,}
  Phys.\ Lett.\ B {\bf 644}, 272 (2007)
  [hep-th/0611091].

\bibitem{Brandhuber:2005jw}
  A.~Brandhuber, S.~McNamara, B.~J.~Spence and G.~Travaglini,
  {\it Loop amplitudes in pure Yang-Mills from generalised unitarity,}
  JHEP {\bf 0510}, 011 (2005)
  [hep-th/0506068].

\bibitem{Ossola:2006us}
  G.~Ossola, C.~G.~Papadopoulos and R.~Pittau,
  {\it Reducing full one-loop amplitudes to scalar integrals at the integrand
  level,}
  Nucl.\ Phys.\ B {\bf 763}, 147 (2007)
  [hep-ph/0609007].

\bibitem{Anastasiou:2006gt}
  C.~Anastasiou, R.~Britto, B.~Feng, Z.~Kunszt and P.~Mastrolia,
  {\it Unitarity cuts and Reduction to master integrals in d dimensions for
  one-loop amplitudes,}
  JHEP {\bf 0703}, 111 (2007)
  [hep-ph/0612277].

\bibitem{Bern:2007dw}
  Z.~Bern, L.~J.~Dixon and D.~A.~Kosower,
  {\it On-Shell Methods in Perturbative QCD,}
  Annals Phys.\  {\bf 322}, 1587 (2007)
  [arXiv:0704.2798 [hep-ph]].

\bibitem{Forde:2007mi}
  D.~Forde,
  {\it Direct extraction of one-loop integral coefficients,}
  Phys.\ Rev.\ D {\bf 75}, 125019 (2007)
  [arXiv:0704.1835 [hep-ph]].

\bibitem{Badger:2008cm}
  S.~D.~Badger,
  {\it Direct Extraction Of One Loop Rational Terms,}
  JHEP {\bf 0901}, 049 (2009)
  [arXiv:0806.4600 [hep-ph]].

\bibitem{Giele:2008ve}
  W.~T.~Giele, Z.~Kunszt and K.~Melnikov,
  {\it Full one-loop amplitudes from tree amplitudes,}
  JHEP {\bf 0804}, 049 (2008)
  [arXiv:0801.2237 [hep-ph]].

\bibitem{Britto:2006fc}
  R.~Britto and B.~Feng,
  {\it Unitarity cuts with massive propagators and algebraic expressions for
  coefficients,}
  Phys.\ Rev.\ D {\bf 75}, 105006 (2007)
  [hep-ph/0612089].

\bibitem{Britto:2007tt}
  R.~Britto and B.~Feng,
  {\it Integral coefficients for one-loop amplitudes,}
  JHEP {\bf 0802}, 095 (2008)
  [arXiv:0711.4284 [hep-ph]].

\bibitem{Bern:2010qa}
  Z.~Bern, J.~J.~Carrasco, T.~Dennen, Y.~-t.~Huang and H.~Ita,
  {\it Generalized Unitarity and Six-Dimensional Helicity,}
  Phys.\ Rev.\ D {\bf 83}, 085022 (2011)
  [arXiv:1010.0494 [hep-th]].

\bibitem{Anastasiou:2006jv}
  C.~Anastasiou, R.~Britto, B.~Feng, Z.~Kunszt and P.~Mastrolia,
  {\it D-dimensional unitarity cut method,}
  Phys.\ Lett.\ B {\bf 645}, 213 (2007)
  [hep-ph/0609191].


\bibitem{Ellis:2007br}
  R.~K.~Ellis, W.~T.~Giele and Z.~Kunszt,
  {\it A Numerical Unitarity Formalism for Evaluating One-Loop Amplitudes,}
  JHEP {\bf 0803}, 003 (2008)
  [arXiv:0708.2398 [hep-ph]].

\bibitem{Berger:2008sj}
  C.~F.~Berger, Z.~Bern, L.~J.~Dixon, F.~Febres Cordero, D.~Forde, H.~Ita,
  D.~A.~Kosower and D.~Maitre,
  {\it An Automated Implementation of On-Shell Methods for One-Loop Amplitudes,}
  Phys.\ Rev.\ D {\bf 78}, 036003 (2008)
  [arXiv:0803.4180 [hep-ph]].

\bibitem{Ossola:2007ax}
  G.~Ossola, C.~G.~Papadopoulos and R.~Pittau,
  {\it CutTools: A Program implementing the OPP reduction method to compute
  one-loop amplitudes,}
  JHEP {\bf 0803}, 042 (2008)
  [arXiv:0711.3596 [hep-ph]].

\bibitem{Mastrolia:2008jb}
  P.~Mastrolia, G.~Ossola, C.~G.~Papadopoulos and R.~Pittau,
  {\it Optimizing the Reduction of One-Loop Amplitudes,}
  JHEP {\bf 0806}, 030 (2008)
  [arXiv:0803.3964 [hep-ph]].

\bibitem{Giele:2008bc}
  W.~T.~Giele and G.~Zanderighi,
  {\it On the Numerical Evaluation of One-Loop Amplitudes: The Gluonic Case,}
  JHEP {\bf 0806}, 038 (2008)
  [arXiv:0805.2152 [hep-ph]].

\bibitem{Berger:2009zg}
  C.~F.~Berger, Z.~Bern, L.~J.~Dixon, F.~Febres Cordero, D.~Forde, T.~Gleisberg,
H.~Ita and D.~A.~Kosower {\it et al.},
  {\it Precise Predictions for $W$ + 3 Jet Production at Hadron Colliders,}
  Phys.\ Rev.\ Lett.\  {\bf 102}, 222001 (2009)
  [arXiv:0902.2760 [hep-ph]].

\bibitem{Badger:2010nx}
  S.~Badger, B.~Biedermann and P.~Uwer,
  {\it NGluon: A Package to Calculate One-loop Multi-gluon Amplitudes,}
  Comput.\ Phys.\ Commun.\  {\bf 182}, 1674 (2011)
  [arXiv:1011.2900 [hep-ph]].

\bibitem{Berger:2010zx}
  C.~F.~Berger, Z.~Bern, L.~J.~Dixon, F.~Febres Cordero, D.~Forde, T.~Gleisberg,
  H.~Ita and D.~A.~Kosower {\it et al.},
  {\it Precise Predictions for W + 4 Jet Production at the Large Hadron
  Collider,}
  Phys.\ Rev.\ Lett.\  {\bf 106}, 092001 (2011)
  [arXiv:1009.2338 [hep-ph]].

\bibitem{Hirschi:2011pa}
  V.~Hirschi, R.~Frederix, S.~Frixione, M.~V.~Garzelli, F.~Maltoni and
  R.~Pittau,
  {\it Automation of one-loop QCD corrections,}
  JHEP {\bf 1105}, 044 (2011)
  [arXiv:1103.0621 [hep-ph]].


\bibitem{Bern:1997nh}
  Z.~Bern, J.~S.~Rozowsky and B.~Yan,
  {\it Two loop four gluon amplitudes in N=4 superYang-Mills,}
  Phys.\ Lett.\ B {\bf 401}, 273 (1997)
  [hep-ph/9702424].

\bibitem{Bern:2000dn}
  Z.~Bern, L.~J.~Dixon and D.~A.~Kosower,
  {\it A Two loop four gluon helicity amplitude in QCD,}
  JHEP {\bf 0001}, 027 (2000)
  [hep-ph/0001001].

\bibitem{Glover:2001af}
  E.~W.~N.~Glover, C.~Oleari and M.~E.~Tejeda-Yeomans,
  {\it Two loop QCD corrections to gluon-gluon scattering,}
  Nucl.\ Phys.\ B {\bf 605}, 467 (2001)
  [hep-ph/0102201].

\bibitem{Bern:2002tk}
  Z.~Bern, A.~De Freitas and L.~J.~Dixon,
  {\it Two loop helicity amplitudes for gluon-gluon scattering in QCD and
  supersymmetric Yang-Mills theory,}
  JHEP {\bf 0203}, 018 (2002)
  [hep-ph/0201161].

\bibitem{Anastasiou:2000kg}
  C.~Anastasiou, E.~W.~N.~Glover, C.~Oleari and M.~E.~Tejeda-Yeomans,
  {\it Two-loop QCD corrections to the scattering of massless distinct quarks,}
  Nucl.\ Phys.\ B {\bf 601}, 318 (2001)
  [hep-ph/0010212].

\bibitem{Anastasiou:2000ue}
  C.~Anastasiou, E.~W.~N.~Glover, C.~Oleari and M.~E.~Tejeda-Yeomans,
  {\it Two loop QCD corrections to massless identical quark scattering,}
  Nucl.\ Phys.\ B {\bf 601}, 341 (2001)
  [hep-ph/0011094].

\bibitem{Anastasiou:2001sv}
  C.~Anastasiou, E.~W.~N.~Glover, C.~Oleari and M.~E.~Tejeda-Yeomans,
  {\it Two loop QCD corrections to massless quark gluon scattering,}
  Nucl.\ Phys.\ B {\bf 605}, 486 (2001)
  [hep-ph/0101304].

\bibitem{Buchbinder:2005wp}
  E.~I.~Buchbinder and F.~Cachazo,
  {\it Two-loop amplitudes of gluons and octa-cuts in N=4 super Yang-Mills,}
  JHEP {\bf 0511}, 036 (2005)
  [hep-th/0506126].

\bibitem{Cachazo:2008vp}
  F.~Cachazo,
  {\it Sharpening The Leading Singularity,}
  arXiv:0803.1988 [hep-th].

\bibitem{Gluza:2010ws}
  J.~Gluza, K.~Kajda and D.~A.~Kosower,
  {\it Towards a Basis for Planar Two-Loop Integrals,}
  Phys.\ Rev.\ D {\bf 83}, 045012 (2011)
  [arXiv:1009.0472 [hep-th]].

\bibitem{Schabinger:2011dz}
  R.~M.~Schabinger,
  {\it A New Algorithm For The Generation Of Unitarity-Compatible Integration By
  Parts Relations,}
  JHEP {\bf 1201}, 077 (2012)
  [arXiv:1111.4220 [hep-ph]].

\bibitem{Kosower:2011ty}
  D.~A.~Kosower and K.~J.~Larsen,
  {\it Maximal Unitarity at Two Loops,}
  Phys.\ Rev.\ D {\bf 85}, 045017 (2012)
  [arXiv:1108.1180 [hep-th]].

\bibitem{Larsen:2012sx}
  K.~J.~Larsen,
  {\it Global Poles of the Two-Loop Six-Point N=4 SYM integrand,}
  Phys.\ Rev.\ D {\bf 86}, 085032 (2012)
  [arXiv:1205.0297 [hep-th]].

\bibitem{CaronHuot:2012ab}
  S.~Caron-Huot and K.~J.~Larsen,
  {\it Uniqueness of two-loop master contours,}
  JHEP {\bf 1210}, 026 (2012)
  [arXiv:1205.0801 [hep-ph]].

\bibitem{Johansson:2012zv}
  H.~Johansson, D.~A.~Kosower and K.~J.~Larsen,
  {\it Two-Loop Maximal Unitarity with External Masses,}
  Phys.\ Rev.\ D {\bf 87}, 025030 (2013)
  [arXiv:1208.1754 [hep-th]].

\bibitem{Johansson:2012sf}
  H.~Johansson, D.~A.~Kosower and K.~J.~Larsen,
  {\it An Overview of Maximal Unitarity at Two Loops,}
  PoS LL {\bf 2012}, 066 (2012)
  [PoS LL {\bf 2012}, 066 (2012)]
  [arXiv:1212.2132 [hep-th]].

\bibitem{Sogaard:2013yga} 
  M.~S{\o}gaard,
  {\it Global Residues and Two-Loop Hepta-Cuts,}
  JHEP {\bf 1309}, 116 (2013)
  [arXiv:1306.1496 [hep-th]].

\bibitem{Johansson:2013sda}
  H.~Johansson, D.~A.~Kosower and K.~J.~Larsen,
  {\it Maximal Unitarity for the Four-Mass Double Box,}
  arXiv:1308.4632 [hep-th].

\bibitem{Badger:2012dp}
  S.~Badger, H.~Frellesvig and Y.~Zhang,
  {\it Hepta-Cuts of Two-Loop Scattering Amplitudes,}
  JHEP {\bf 1204}, 055 (2012)
  [arXiv:1202.2019 [hep-ph]].

\bibitem{Mastrolia:2011pr}
  P.~Mastrolia and G.~Ossola,
  {\it On the Integrand-Reduction Method for Two-Loop Scattering Amplitudes,}
  JHEP {\bf 1111}, 014 (2011)
  [arXiv:1107.6041 [hep-ph]].

\bibitem{Badger:2012dv}
  S.~Badger, H.~Frellesvig and Y.~Zhang,
  {\it An Integrand Reconstruction Method for Three-Loop Amplitudes,}
  JHEP {\bf 1208}, 065 (2012)
  [arXiv:1207.2976 [hep-ph]].

\bibitem{Zhang:2012ce}
  Y.~Zhang,
  {\it Integrand-Level Reduction of Loop Amplitudes by Computational Algebraic
  Geometry Methods,}
  JHEP {\bf 1209}, 042 (2012)
  [arXiv:1205.5707 [hep-ph]].

\bibitem{Badger:2013gxa} 
  S.~Badger, H.~Frellesvig and Y.~Zhang,
  {\it A Two-Loop Five-Gluon Helicity Amplitude in QCD,}
  arXiv:1310.1051 [hep-ph].

\bibitem{Feng:2012bm}
  B.~Feng and R.~Huang,
  {\it The classification of two-loop integrand basis in pure four-dimension,}
  JHEP {\bf 1302}, 117 (2013)
  [arXiv:1209.3747 [hep-ph]].

\bibitem{Mastrolia:2012an}
  P.~Mastrolia, E.~Mirabella, G.~Ossola and T.~Peraro,
  {\it Scattering Amplitudes from Multivariate Polynomial Division,}
  Phys.\ Lett.\ B {\bf 718}, 173 (2012)
  [arXiv:1205.7087 [hep-ph]].

\bibitem{Mastrolia:2012wf}
  P.~Mastrolia, E.~Mirabella, G.~Ossola and T.~Peraro,
  {\it Integrand-Reduction for Two-Loop Scattering Amplitudes through Multivariate
  Polynomial Division,}
  arXiv:1209.4319 [hep-ph].

\bibitem{Mastrolia:2012du}
  P.~Mastrolia, E.~Mirabella, G.~Ossola, T.~Peraro and H.~van Deurzen,
  {\it The Integrand Reduction of One- and Two-Loop Scattering Amplitudes,}
  PoS LL {\bf 2012} (2012) 028
  [arXiv:1209.5678 [hep-ph]].

\bibitem{Kleiss:2012yv}
  R.~H.~P.~Kleiss, I.~Malamos, C.~G.~Papadopoulos and R.~Verheyen,
  {\it Counting to One: Reducibility of One- and Two-Loop Zmplitudes at the
  Integrand Level,}
  JHEP {\bf 1212}, 038 (2012)
  [arXiv:1206.4180 [hep-ph]].

\bibitem{Huang:2013kh} 
  R.~Huang and Y.~Zhang,
  {\it On Genera of Curves from High-loop Generalized Unitarity Cuts,}
  JHEP {\bf 1304}, 080 (2013)
  [arXiv:1302.1023 [hep-ph]].

\bibitem{Mastrolia:2013kca} 
  P.~Mastrolia, E.~Mirabella, G.~Ossola and T.~Peraro,
  arXiv:1307.5832 [hep-ph].


\bibitem{Bern:2009xq}
  Z.~Bern, J.~J.~M.~Carrasco, H.~Ita, H.~Johansson and R.~Roiban,
  {\it On the Structure of Supersymmetric Sums in Multi-Loop Unitarity Cuts,}
  Phys.\ Rev.\ D {\bf 80}, 065029 (2009)
  [arXiv:0903.5348 [hep-th]].

\bibitem{Sogaard:2011pr}
  M.~Sogaard,
  {\it Supersums for all supersymmetric amplitudes,}
  Phys.\ Rev.\ D {\bf 84}, 065011 (2011)
  [arXiv:1106.3785 [hep-th]].

\bibitem{ArkaniHamed:2009dn}
  N.~Arkani-Hamed, F.~Cachazo, C.~Cheung and J.~Kaplan,
  {\it A Duality For The S Matrix,}
  JHEP {\bf 1003}, 020 (2010)
  [arXiv:0907.5418 [hep-th]].

\bibitem{Bern:2008qj}
  Z.~Bern, J.~J.~M.~Carrasco and H.~Johansson,
  {\it New Relations for Gauge-Theory Amplitudes,}
  Phys.\ Rev.\ D {\bf 78}, 085011 (2008)
  [arXiv:0805.3993 [hep-ph]].

\bibitem{Bern:2010ue}
  Z.~Bern, J.~J.~M.~Carrasco and H.~Johansson,
  {\it Perturbative Quantum Gravity as a Double Copy of Gauge Theory,}
  Phys.\ Rev.\ Lett.\  {\bf 105}, 061602 (2010)
  [arXiv:1004.0476 [hep-th]].

\bibitem{Smirnov:1999gc}
  V.~A.~Smirnov,
  {\it Analytical result for dimensionally regularized massless on shell double
  box,}
  Phys.\ Lett.\ B {\bf 460}, 397 (1999)
  [hep-ph/9905323].

\bibitem{Smirnov:1999wz}
  V.~A.~Smirnov and O.~L.~Veretin,
  {\it Analytical results for dimensionally regularized massless on-shell double
  boxes with arbitrary indices and numerators,}
  Nucl.\ Phys.\ B {\bf 566}, 469 (2000)
  [hep-ph/9907385].

\bibitem{Tausk:1999vh}
  J.~B.~Tausk,
  {\it Nonplanar massless two loop Feynman diagrams with four on-shell legs,}
  Phys.\ Lett.\ B {\bf 469}, 225 (1999)
  [hep-ph/9909506].

\bibitem{Anastasiou:2000mf}
  C.~Anastasiou, T.~Gehrmann, C.~Oleari, E.~Remiddi and J.~B.~Tausk,
  {\it The Tensor reduction and master integrals of the two loop massless crossed
  box with lightlike legs,}
  Nucl.\ Phys.\ B {\bf 580}, 577 (2000)
  [hep-ph/0003261].

\bibitem{Smirnov:2003vi} 
  V.~A.~Smirnov,
  {\it Analytical result for dimensionally regularized massless on shell planar
  triple box,}
  Phys.\ Lett.\ B {\bf 567}, 193 (2003)
  [hep-ph/0305142].

\bibitem{Bern:2005iz} 
  Z.~Bern, L.~J.~Dixon and V.~A.~Smirnov,
  {\it Iteration of planar amplitudes in maximally supersymmetric Yang-Mills
  theory at three loops and beyond,}
  Phys.\ Rev.\ D {\bf 72}, 085001 (2005)
  [hep-th/0505205].

\bibitem{Smirnov:2013dia} 
  A.~V.~Smirnov and V.~A.~Smirnov,
  arXiv:1302.5885 [hep-ph].

\bibitem{vonManteuffel:2012np} 
  A.~von Manteuffel and C.~Studerus,
  arXiv:1201.4330 [hep-ph].

\bibitem{M2}
D.~R. Grayson and M.~E. Stillman, {\it Macaulay2, a software system for research
  in algebraic geometry.} Available at http://www.math.uiuc.edu/Macaulay2/. 

\bibitem{MR0463157}
 R.~Hartshorne, {\sl Algebraic geometry}. Springer-Verlag, New York,
 1977. Graduate Texts in Mathematics, No. 52.

\bibitem{MR507725}
 P.~Griffiths, J.~Harris, {\sl Principles of Algebraic geometry}.
 Wiley-Interscience [John Wiley \& Sons], New York,
 1978. 
\end{thebibliography}
\end{document}